\documentclass[10pt,twoside,a4paper]{article}

\usepackage[english]{babel}
\usepackage[utf8]{inputenc}

\usepackage{amsfonts,amsthm,amssymb,amsmath,amsthm,latexsym,wasysym,stmaryrd,mathrsfs,dsfont,txfonts}
\usepackage{multirow,rotating,subfigure,graphicx,bigstrut,lscape,multicol,slashbox,fancyhdr,enumitem,accents,mathtools,relsize,exscale,picins}

\usepackage{pdftricks}
\usepackage{color}
\usepackage[pdftex,all]{xy}

\usepackage{hyperref}
\usepackage[centerlast]{caption2}

\graphicspath{{Figures/}}

\theoremstyle{theorem}
\newtheorem {theo}{Theorem}[section]
\newtheorem*{theo*}{Theorem}
\newtheorem {lemme}[theo]{Lemma}
\newtheorem*{lemme*}{Lemma}
\newtheorem {prop}[theo]{Proposition}
\newtheorem*{prop*}{Proposition}
\newtheorem {cor}[theo]{Corollary}
\newtheorem*{cor*}{Corollary}

\newtheorem*{cor_proof*}{Corollary (of the proof)}

\newtheorem*{conjecture*}{Conjecture}
\theoremstyle{definition}
\newtheorem {defi}{Definition}[section]
\newtheorem*{defi*}{Definition}
\newtheorem {nota}[defi]{Notation}
\newtheorem*{nota*}{Notation}
\theoremstyle{remark}
\newtheorem {remarque}{Remark}[section]
\newtheorem*{remarque*}{Remark}
\newtheorem {remarques}[remarque]{Remarks}
\newtheorem*{remarques*}{Remarks}

\newtheorem*{convention*}{Convention}

\newtheorem*{exemple*}{Example}

\newtheorem*{exemples*}{Examples}
\newtheorem {question}[remarque]{Question}
\newtheorem*{question*}{Question}

\newtheorem*{fact*}{Fact}

\def\e{\varepsilon}
\def\ee{{\underline{\e}}}

\def\p{\partial}

\def\BB{{\mathcal B}}
\def\C{{\mathds C}}

\def\E{{\mathds E}}
\def\tE{{\widetilde{\mathds E}}}
\def\F{{\mathds F}}

\def\tG{{\widetilde{G}}}
\def\GG{{\mathcal G}}
\def\tGG{{\widetilde{\mathcal G}}}

\def\HH{\mathcal{H}}

\def\N{{\mathds N}}

\def\R{{\mathds R}}

\def\Z{{\mathds Z}}

\def\2Z{{\fract{\Z}/{2\Z}}}
\def\tX{\widetilde{X}}
\def\tY{\widetilde{Y}}
\def\tZ{\widetilde{Z}}

\def\Cone{{\textnormal{Cone}}}

\def\Fix{{\textnormal{Fix}}}
\def\Kh{{\textnormal{Kh}}}
\def\Hom{{\textnormal{Hom}}}
\def\Id{{\textnormal{Id}}}
\def\Im{{\textnormal{Im}}}

\def\Ker{{\textnormal{Ker}}}
\def\Mat{{\textnormal{Mat}}}

\def\arr{{\textnormal{round}}}

\def\ie{{\it i.e. }}

\def\rk{{\textnormal{rk}}}

\newcommand{\Duk}[1]{{D^\textrm{uk}_{#1}}}
\newcommand{\Dul}[1]{{D^\textrm{ul}_{#1}}}
\newcommand{\Dtl}[1]{{D^\textrm{tl}_{#1}}}

\def\pcup{\operatornamewithlimits{\cup}\limits}
\def\poplus{\operatornamewithlimits{\oplus}\limits}
\def\ppoplus{\mathlarger{\mathlarger{\mathlarger{\mathlarger{\operatornamewithlimits{\oplus}\limits}}}}}

\newcommand{\noi}{\noindent}
\newcommand{\disp}{\displaystyle}

\def\fract#1/#2{\hbox{\leavevmode
  \kern.1em \raise .25ex \hbox{\the\scriptfont0 $#1$}\kern-.1em }\big/
  {\hbox{\kern-.15em \lower .5ex \hbox{\the\scriptfont0 $#2$}} }}

\def\fractt#1/#2{\hbox{\leavevmode
  \kern.1em \raise .25ex \hbox{\the\scriptfont0 $#1$}\kern-.1em
}\lower .2ex\hbox{\Big/}
  {\hbox{\kern-.15em \lower .8ex \hbox{\the\scriptfont0 $#2$}} }}

\def\subfract#1/#2{\hbox{\leavevmode
  \kern.1em \raise .25ex \hbox{\the\scriptfont0 \scriptsize $#1$}\kern-.1em }/
  {\hbox{\kern-.15em \lower .5ex \hbox{\the\scriptfont0 \scriptsize $#2$}} }}

\newcommand{\dessin}[2]{
  \vcenter{\hbox{\includegraphics[height=#1]{#2.pdf}}}}

\newcommand{\func}[3]{
  #1 \colon #2 \longrightarrow #3}

\newcommand{\qb}[1]{
  |#1\rangle}

\begin{document}

\title{An application of Khovanov homology to quantum codes}
\author{Benjamin \textsc{Audoux}}
\date{\today}

\maketitle

\begin{abstract}
We use Khovanov homology to define families of LDPC quantum
error-correcting codes: unknot codes with
asymptotical parameters $\left\llbracket \frac{3^{2 \ell+1}}{\sqrt{8\pi
      \ell}};1;2^\ell\right\rrbracket$; unlink codes with asymptotical
parameters $\left\llbracket \sqrt{\frac{3}{2\pi
      \ell}}6^\ell;2^\ell;2^\ell\right\rrbracket$ and $(2,\ell)$--torus link codes
with asymptotical parameters $\llbracket n;1;d_n\rrbracket$ where $d_n>\frac{\sqrt{n}}{1.62}$. 
\end{abstract}

\section*{Introduction}
\label{sec:introduction}

Classical error--correcting codes have been now studied for decades.
Among them, some codes (\cite{Gallager}), defined by sparse matrices and called LDPC (Low Density Parity Check), noteworthily come with fast decoding algorithms.
Since the end of the last century, error--correcting codes for quantum
computing were also known to exist and explicit constructions were given.
A. R. Calderbank, P. Shor and
A. Steane (\cite{Calderbank},\cite{Steane}) described, for instance, a way
to associate such a code to any pair $(\mathbf{H}_X,\mathbf{H}_Z)$ of
$\F_2$--matrices with $\mathbf{H}_X \mathbf{H}_Z^t=0$.
This procedure allows the construction of several codes with good
parameters ; it means infinite families of quantum codes whose dimension (usually denoted by $k$) and number of
rectifiable errors (which is related to the minimum distance, usually denoted by
$d$) are both linear in the length of codewords (usually denoted by $n$).

However, quickness in quantum decoding is
all the more crucial since corrections should
occur as fast as quantum decoherence arises.
It is then natural to try to transpose the LDPC notion for classical codes into a
quantum counterpart, looking for pairs of matrices $(\mathbf{H}_X,\mathbf{H}_Z)$ with
minimally weighted rows.
Surprisingly, topology appeared to be a fruitful field for such
a project.
This was initiated by Kitaev codes (\cite{Kitaev}) who
defined such a family of, so-called toric, codes by considering a $m\times m$--squared tesselation
of the $S^1\times S^1$--torus.
It led to codes with parameters equal to $\llbracket
n;k;d\rrbracket=\llbracket
2m^2;2;m\rrbracket$.
Toric codes were then generalized to surface (\cite{BD1}) and color (\cite{BD2}) codes. Other LDPC quantum codes were also defined;
see for instance the
constructions given by M. Freedman, D. Meyer and F. Luo in
\cite{Freedman} with asymptotical parameters $\llbracket
n;a\!\sqrt{n};b\!\sqrt{n}\ln(n)\rrbracket$ or by
J.-P. Tillich and G. Zemor in \cite{TZ} with asymptotical parameters $\llbracket
n;cn;d\!\sqrt{n}\rrbracket$, where $a$, $b$,
$c$ and $d$ are some constants.
It is striking that none of these, and even none of any known LDPC quantum error-correcting codes
families, has a minimum distance $d$ that grows faster than
$n^{\alpha}$ for any $\alpha>\subfract1/2$.
It is still an open question
to know whether there is actually a general square root
barrier for minimum distance in LDPC
quantum codes or if this is only due to an ``excess of structure'' in
these constructions.
Indeed, constructing LDPC quantum codes remains challenging, and the
few examples which are known to date carry lots of structure --- in
particular, a duality structure --- and symmetry. This enables exact
comptutation of parameters but may yield artificial restrictions.
The square root barrier has been proved for surfaces and color codes
(\cite{Delfosse13},\cite{Fetaya}). There is
thus a need for new constructions. 

In this paper, we explore a new side of topology which is likely to hold interesting
quantum codes.
Khovanov homology is a link invariant defined in \cite{Khovanov}. To
any diagram representation of a link, it associates a chain complex
whose homology depends on the underlying link only.
The chain complex
is actually bigraded and its Euler characteristic is famed for categorifying the
Jones polynomial, however we will not be interested here in this second
non homological grading.
Khovanov homology has a rich structure, in particular a Poincaré
duality property, that makes easier the computation of minimum distances.
As a matter of fact, we study three families of codes, associated to
some very simple knots and links, and compute explicitely their
parameters. Asymptotically, we respectively obtain $\left\llbracket \frac{3^{2 \ell+1}}{\sqrt{8\pi \ell}};1;2^\ell\right\rrbracket$, $\left\llbracket \sqrt{\frac{3}{2\pi \ell}}6^\ell;2^\ell;2^\ell\right\rrbracket$ and $\llbracket
n;1;a\!\sqrt{n}\rrbracket$ with $a$ a constant.
This is below the
parameters of Freedman--Meyer--Luo and Tillich--Zemor codes, but
reaches, and even beats, toric codes and most other known
ones. Moreover, there are still many others candidates among link
diagrams to look at and
other codes properties to study such as minimal amount of energy needed to
reach an unrectifiable error.
Moreover, it is worthwhile to note that, even if the construction
drastically differs from its predecessors, it seems to run into the
same square root bound for minimum distance.
Finally, even if this study was initially motivated by quantum computing
interests, it opens some questions (see {\it e.g.} question \ref{quest:ConjReid})
that may result on interesting properties of Khovanov homology, even
from the knot theory point of view.

This paper aims at being readable by both topologists and code
theorists. It begins by a review of LDPC CSS codes followed by a
review of chain
complexes and homology. The first part ends with a generic way to define one
of the former using the latter.
The second part is devoted to the definition of Khovanov homology and
to some of its properties.
Third, fourth and fifth parts deal each with a family of codes
associated, respectively, to diagrams of the unknot, of the unlinks and
of the $(2,n)$--torus knots and links. All the parameters of the codes
are computed
there.
Finally, in order to lighten the core of the text, a technical
appendix gathers some analytical proofs needed on the way.

 \subsection*{Acknowledgement}
 The author thanks Alain Couvreur and Gilles Zemor for
 introducing him to quantum codes and to their connection with topology.
 He is also deeply grateful to Nicolas Delfosse for answering all his
 (numerous) questions on quantum computing. He finally wants to thank Rinat
 Kashaev for a simplification in the proof of
 Prop. \ref{prop:SumCarre}.
 The author is supported by ANR project VasKho and CNRS PEPS project TOCQ.



\section{Chain complex codes}
\label{sec:code-chain-complexes}

\subsection{From quantum errors to codes}
\label{sec:from-quantum-error}

For more details, the author recommands \cite{Nielsen}, \cite{Preskill} or
the (french) introduction of \cite{Delfosse} to the reader.
This section is a rough overview of error--correcting quantum codes
adressed to non specialists.

\subsubsection{Qubits and their errors}
\label{sec:qubits}

In quantum theory, the elementary piece of information is the qubit.
It is a unitary element in the $\C$--vector space
$\HH$ spanned by two generators, usually denoted by $\qb{0}$ and $\qb{1}$.
We denote the space of qubits by $\HH_1$.
Actually, only the images in the projective quotient can be physically apprehended, but since it
will be fruitful to deal with signs issues, we will often switch between the (non commutative) affine
and the (commutative) projective cases. For convenience, we will use
notation with tildas each time we deal with affine elements.

Unlike the classical case, multiple qubits do not just concatenate:
they can entangle. From the postulates of quantum mechanics, $n$ qubits are described by unitary
elements in $\HH^{\otimes n}$; they are of the form
$\disp{\sum_{x\in\{0,1\}^n}}\alpha_x\qb{x}$ with
$\sum_x|\alpha_x|^2=1$.
We denote the space of such $n$--qubits by $\HH^n_1$.

Transmitting, or even just keeping stored, a $n$--qubit may alter it.
On a single qubit, a set of possible alterations is the Pauli group $\tGG_1$,
generated by three elements:
\[
\tX:\begin{array}{ccr}
  \qb{0} & \mapsto & \qb{1}\\
  \qb{1} & \mapsto & \qb{0}\\
\end{array},
\hspace{1cm}
\tY:\begin{array}{ccr}
  \qb{0} & \mapsto & -i\qb{1}\\
  \qb{1} & \mapsto & i\qb{0}\\  
\end{array},
\hspace{1cm}
\tZ:\begin{array}{ccr}
\qb{0} & \mapsto & \qb{0}\\
  \qb{1} & \mapsto & -\qb{1}\\

\end{array}.
\]
Of course, they are not the only errors which may occur, but they are
an orthogonal basis for them.
For this reason, it is sufficient to focus our effort on them.
We can note that every such Pauli error is of the form $\e A$ with
$\e\in S:=\{\pm1,\pm i\}$ and $A\in\tE:=\{I,\tX,\tY,\tZ\}$ and that any two errors
always do commute or anti-commute.
We denote by $\GG_1$ the projective quotient of $\tGG_1$. It is
an abelian group which is generated by only two elements, for instance
$X$ and $Z$, the images of $\tX$ and $\tZ$.
On a $n$--qubit, every factor can be altered by an error.
The group $\tGG_n=\tGG_1^{\otimes n}$, defined as the set $\tE^n\times
S$ with the obvious product, forms an orthogonal basis for errors on $n$--qubits. 
Here again, every two elements do commute or anti-commute; 
and the projective quotient $\GG_n$ of $\tGG_n$
is $\E^n$, where $\E:=\{I,X,Z,XZ\}$.
The group $\GG_n$ is abelian but we say that two elements commute
(resp. anti-commute) only if their lifts in $\tGG_n$ do commute
(resp. anti-commute).
Note that it does not depend on the choosen lifts.

\subsubsection{CSS codes}
\label{sec:css-codes}

A quantum code $C$ of length $n\in\N^*$ and dimension
$k\in\llbracket1,n\rrbracket$ is a $2^k$--dimensional subspace of
$\HH^{\otimes n}$.
It makes possible the storage of a $k$--qubit in the form of a
$n$--qubit, what enables, as we will see, a correction process for
small alterations of the encoding $n$--qubits.
The terminology, here, may be misleading since the dimension of a
quantum code refers to the number of encoded qubits and not to the
actual dimension of the code as a $\C$--vector space.
We define a \emph{codeword} as any element of $C$.

Let $G$ be a subgroup of $\GG_n$ such that $G$ is liftable to a group $\tG\subset\tGG_n$.
For every $g\in G$, we denote by $\widetilde{g}$ its lift in $\tG$.
We define $C_G$ as
$\Fix_{\tG}(\HH^n_1):=\{x\in\HH^n_1\ |\ \forall \widetilde{g}\in \tG, \widetilde{g}(x)=x\}$.
Note that it only depends on $G$ and not on the choosen lift $\tG$.
If $G$ is generated by $(n-k)$ independant elements of $\GG_n$, then
one can prove that $C_G$ is a code, so-called \emph{stabilizer code}, of dimension $k$.

We say that $C_G$ is a \emph{CSS\footnote{for Calderbank, Shor and
Steane} code} if $G$ is even more restrictively generated by elements
in $\E_X^n\cup \E_Z^n$ with
$\E_X:=\{1,X\}$ and $\E_Z:=\{1,Z\}$.
Since $\E_X^n$ and $\E_Z^n$ are both abelian and made of order 2
elements, they are both isomorphic to $\F_2^n$.
As a matter of fact, such a set of generators can be described as the rows of two matrices
$\mathbf{H}_X,\mathbf{H}_Z\in\pcup_{p\in\N^*}\Mat_{\F_2}(p,n)$: to a row
$(a_1,\cdots,a_n)\in\F_2^n$ of $A_\alpha$ with $\alpha=X$ or $Z$, we
associate $(\alpha^{a_1},\cdots,\alpha^{a_n})\in \E_\alpha^n$.

The fact that $G$ is liftable in $\tGG_n$ means that every two
generators $x$ and $y$ commute.
Of course, if $x,y\in\E^n_X$ or $x,y\in\E^n_Z$, this is trivially satisfied;
but since $\tX$ and $\tZ$ anticommute, $x\in\E^n_X$ and
$y\in\E^n_Z$ do commute iff they share an even number of non-zero
entries, that is if the product of the associated rows in $\mathbf{H}_X$ and
in $\mathbf{H}_Z$ is zero.
In short, $G$ is liftable iff $\mathbf{H}_X \mathbf{H}^t_Z=0$.

Finally, generators in $\E^n_X$ are necessarily independant from those in
$\E^n_Z$, so the minimal number of independant generators for $G$ is
$\rk(\mathbf{H}_X)+\rk(\mathbf{H}_Z)$.
As a matter of fact, two matrices $\mathbf{H}_X$ and $\mathbf{H}_Z$ such that $\mathbf{H}_X\mathbf{H}_Z^t=0$ being given,
the length $n$ of the associated CSS code is their common number of
columns, and the dimension is $k=n-\rk(\mathbf{H}_X)-\rk(\mathbf{H}_Z)$.

\subsubsection{Decoding and minimum distance}
\label{sec:decoding}

In quantum physics, certain measurements can be seen as orthogonal
projections.
More precisely, for a given orthogonal decomposition
$\HH^n=\overset{\bot}{\oplus}V_i$, there is an associated measure
which sends a unitary element $\sum x_i\in\HH^n_1$ to $\frac{1}{||x_{i_0}||}x_{i_0}$ with
probability $||x_{i_0}||^2$.

Now, let $C_G$ be a CSS code and $\{E_1,\cdots,E_{n-k}\}$ be a minimal set of $n-k$ generators
for $G$.
For every $\sigma:=(s_1,\cdots,s_{n-k})\in\F_2^{n-k}$, we set
$C(\sigma):=\{x\in\HH_1^n\ |\ \forall i\in\llbracket1,n-k\rrbracket, \widetilde{E}_i(x)=(-1)^{s_i}x\}$.
For every error $E\in\GG^n$, we define its syndrome
$\sigma(E):=\big(s_1(E),\cdots,s_{n-k}(E)\big)\in\F_2^{n-k}$ by
$s_i(E)=0$ iff $E$ commutes with $E_i$.
We can note that if $x\in C_G$ and $E\in\GG_n$, then $\widetilde{E}(x)\in
C(\sigma(E))$.
The \emph{weight} of an error is the number of qubits it alters.
For every $\sigma\in\F_2^{n-k }$, we choose a minimally weighted error $E_\sigma$ of
syndrome $\sigma$.

The decomposition
$\HH^n=\overset{\bot}{\ppoplus_{\sigma\in\F_2^{n-k}}}C(\sigma)$ holds
and the associated measure discretizes the
set of possible alterations of a codeword.
Indeed, let $e(x_0)$ be a codeword $x_0\in C_G=\Fix_{\tG}(\HH^n_1)$ altered by
an error $e$ and let assume that the measure projects it to $E(x_0)$ where
$E$ is a Pauli error
of syndrome $\sigma_E$.
Then one can try to correct the error by computing
$\overline{x_0}:=\widetilde{E}_{\sigma_E}\widetilde{E}(x_0)$. By
construction, $\widetilde{E}_{\sigma_E}\widetilde{E}$ has a syndrome
equal to zero, so it commutes with all elements in $G$. If
it is
actually in $G$, then $\overline{x_0}=x_0$ and we got back the initial
codeword.
However, it may happen that $\widetilde{E}_{\sigma_E}\widetilde{E}$ does not belong to
$G$. Then the decoding process fails.

The \emph{minimum distance} of a code is the minimal weight of a non detectible error that
does alter codewords.
For a CSS code $C_G$, it is the minimal
weight of an error which commutes with all the elements of $G$ but does
not belong to $G$.
It corresponds, as we will see in the proof of Prop. \ref{prop:ChainCodes}, to the minimal weight of a
vector which is in the kernel of one of the matrices $\mathbf{H}_X$ or
$\mathbf{H}_Z$ without being spanned by the rows of the other.

\begin{nota}
  For any code, we denote its parameters by $\llbracket n;k;d\rrbracket$ where $n$ is the
  length of the code, $k$ its dimension and $d$ its minimum distance.
\end{nota}

\subsection{From codes to chain complexes}
\label{sec:from-codes-chain}

For further details, the reader can refer to \cite{Weibel},
\cite{Hilton}, \cite{MacLane} or \cite{Lang}.

\subsubsection{Homology and cohomology}
\label{sec:chains-complexes}

Before relating them to quantum codes, we recall some basic
definitions on chain complexes. We will focus here on $\F_2$, but up
to signs issues,
everything remains true for any field. Most of it remains even true
for any ring.

\begin{defi}
  An increasing (resp. decreasing) chain complex $C$ is a $\Z$--graded $\F_2$--vector space
  $\poplus_{i\in\Z} C^i$ (resp. $\poplus_{i\in\Z} C_i$) together with a
  linear map $\func{\p}{C}{C}$ which increases (resp. decreases) the
  grading by one and satisfies $\p^2\equiv 0$.
It is often denoted as
\[
\xymatrix{\cdots \ar[r]^(.45){\p} & C^i \ar[r]^(.4){\p} & C^{i+1} \ar[r]^(.45){\p} & \cdots}.
\]
The grading is called \emph{homological grading}.
If $C$ is non zero for only a finite number of homological degrees,
then we omit all the redundant zero spaces.
\end{defi}

\begin{remarque}
  Unless otherwise specified, chain complexes will be assumed
  to be increasing.
  This convention is opposite to the usual one, but it sticks to the standard
  appellation ``Khovanov homology'', which should be more appropriately called
  ``Khovanov cohomology''.
\end{remarque}

\begin{defi}
   If $C:=\Big(\poplus_{i\in\Z} C^i,\p \Big)$ is a chain complex, then
   its dual $C^\vee$ is the decreasing chain complex $\Big(\poplus_{i\in\Z}
   C^\vee_i,\p^\vee \Big)$ defined, for every $i\in\Z$, by $C^\vee_i=\Hom(C^i,\F_2)$ and
   $\big(\p^\vee(f)\big)(c)=f\big(\p(c)\big)$ for every $f\in C^\vee_i$ and $c\in C^{i-1}$.
\end{defi}

\begin{defi}
  If $C:=\Big(\poplus_{i\in\Z} C^i,\p \Big)$ is a chain complex, then
  its homology $H^*(C)$ is the graded space
   $\poplus_{i\in\Z} H^i(C):=\poplus_{i\in\Z} \fractt{\Big(\Ker(\p)
     \cap C^i\Big)}/{\Big(\Im(\p) \cap C^i\Big)}$ and
   its cohomology $H_*(C)$ the graded space
  $\poplus_{i\in\Z} H_i(C):=\poplus_{i\in\Z}\fractt{\Big(\Ker(\p^\vee)
    \cap C^\vee_i\Big)}/{\Big(\Im(\p^\vee) \cap C^\vee_i\Big)}$
   where $C^\vee$ is the dual of $C$.
   
   For every $x\in\Ker(\p)$ (resp. $x\in\Ker(\p^\vee)$), we denote by
   $[x]$ its image in $H^*(C)$ (resp. $H_*(C)$).
\end{defi}

Now, we prove a very elementary lemma which will be central in the proof
of Prop. \ref{prop:MinDist}.
\begin{lemme}\label{lem:CoHoNull}
  Let $C:=\Big(\poplus_{i\in\Z} C^i,\p \Big)$ be a chain complex, $r$
  an integer and
  $\{\alpha_i\}_{i\in I}\subset \Ker(\p)\cap C^r$ a finite set such
  that $\big\{[\alpha_i]\big\}_{i\in I}$ generates $H^r(C)$.
  Then every $\varphi\in\Ker(\p^\vee)\cap C^\vee_{r}$ satisfying
  $\varphi(\alpha_i)=0$ for every $i\in I$ is null in $H_r(C)$.
\end{lemme}
\begin{proof}
  Since $\big\{[\alpha_i]\big\}_{i\in I}$ generates $H^r(C)$, every
  $x\in \Ker(\p)\cap C^r$ can be written $x=\sum_{i\in I'\subset
    I}\alpha_i +\p(y)$ with $y\in C^{r-1}$.
  Then
  $\varphi(x)=\varphi\big(\p(y)\big)=\big(\p^*(\varphi)\big)(y)=0$ and
  $\varphi_{|\Ker(\p)}\equiv0$.
  Now, consider a basis $\{\beta_j\}_{j\in J}$ of $\Ker(\varphi)^\perp\subset\Ker(\p)^\perp$ in
  $C^r$, set $\beta'_j=\p(\beta_j)\neq0$ for all $j\in J$ and define
  $g\in\Hom(C^{r+1},\F_2)$ by $g(\beta'_j)=\varphi(\beta_j)$ for all $j\in J$
  and $g_{|\F ^\perp_2<\beta'_j>}\equiv 0$.
  Then $\varphi=g\circ\p\in\Im(\p^\vee)$ and $[\varphi]=0$.
  \end{proof}

\subsubsection{Operations on chain complexes}
\label{sec:oper-chain-compl}

Later on the paper, we will need the following definitions and propositions.

\begin{defi}
  If $C_1:=\Big(\poplus_{i\in\Z} C_1^i,\p_1 \Big)$ and $C_2:=\Big
  (\poplus_{i\in\Z} C_2^i,\p_2 \Big)$ are two chain complexes, then
  $C_1\otimes C_2$ is the chain complex $\Big (\poplus_{i\in\Z} C^i,\p
  \Big)$ defined by $C^i=\poplus_{j\in\Z}\big(C_1^j\otimes
  C_2^{i-j}\big)$ and $\p(c_1\otimes c_2)=\p_1(c_1)\otimes
  c_2+c_1\otimes\p_2(c_2)$ for every $c_1\in C_1$ and $c_2\in C_2$.
\end{defi}

\begin{prop}[Künneth formula]
  If $C_1$ and $C_2$ are two chain complexes, then $H^*(C_1\otimes
  C_2)\cong H^*(C_1)\otimes H^*(C_2)$ and $H_*(C_1\otimes C_2)\cong H_*(C_1)\otimes H_*(C_2)$ as graded spaces.
\end{prop}

\begin{defi}
  If $C_1:=\Big(\poplus_{i\in\Z} C_1^i,\p_1 \Big)$ and $C_2:=\Big
  (\poplus_{i\in\Z} C_2^i,\p_2 \Big)$ are two chain complexes, then
  $f:=\big(\func{f^i}{C_1^i}{C_2^i}\big)_{i\in\Z}$ is a chain map iff it commutes with
  the differentials, \ie iff $\p_2\circ f=f\circ\p_1$.

  The \emph{cone} of $f$ is the chain complex
  $\Cone(f):=\Big(\poplus_{i\in\Z} C^i,\p \Big)$ defined by
  $C^i:=C_1^i\oplus C_2^{i-1}$ for every $i\in\Z$ and $\p=\left(
    \begin{array}{cc}
      \p_1 & 0\\f&\p_2
    \end{array}
\right)$.
\end{defi}
\begin{prop}
  A chain map $\func{f}{C_1}{C_2}$ between two chain complexes $C_1$
  and $C_2$ induces maps at the level of homology and cohomology which
  are denoted by $\func{f^*}{H^*(C_1)}{H^*(C_2)}$ and $\func{f_*}{H_*(C_1)}{H_*(C_2)}$
\end{prop}

\subsubsection{Exact sequences}
\label{sec:exact-sequences}

The following notion will be usefull to compute homologies.
\begin{defi}
  An exact sequence is a chain complex $(C,\p)$ with
  homology equal to zero in all degrees.
  It means that $\Ker(\p)=\Im(\p)$.
\end{defi}
\begin{prop}
  If $(C_0,\p_0)$, $(C_1,\p_1)$ and $(C_2,\p_2)$ are three chain
  complexes such that, for every $n\in\Z$, there are maps
  $\func{\iota_n}{C_0^n}{C_1^n}$ and $\func{\pi_n}{C_1^n}{C_2^n}$
  which commute with the differentials $\p_0$, $\p_1$ and $\p_2$ and
  such that 
\[
\xymatrix{0\ar[r] & C^n_0 \ar[r]^{\iota_n} & C^n_1 \ar[r]^{\pi_n} & C^n_2 \ar[r] &
  0}
\]
\noi is an exact sequence, then
\[
\xymatrix{\cdots \ar[r]^(.4){f^*_{n-1}} & H^n(C_0) \ar[r]^(.45){\iota^*_n} & H^n(C_1)
\ar[r]^{\pi^*_n} & H^n(C_2) \ar[r]^(.45){f^*_n} & H^{n+1}(C_0)
\ar[r]^(.55){\iota^*_{n+1}} & \cdots}
\]
\noi is an exact sequence , where, for all $n\in\Z$, $\iota_n^*$ and $\pi^*_n$ are the maps
induced in homology by $\iota_n$ and $\pi_n$ and $f^*_n$ is some
connecting map.
\end{prop}
\begin{remarque}
  The condition on the short exact sequence just states that maps
  $\iota_n$ are injective, maps $\pi_n$ are surjective and $\Ker(\pi_n)=\Im(\iota_n)$.
\end{remarque}
\begin{prop}
  If $\func{f}{C_1}{C_2}$ is a chain map, then $\Cone(f):=\poplus_{i\in\Z} C^i$ fits the
  following short exact sequence in every degree $n\in\N$:
\[
\xymatrix{0\ar[r] & C^{n-1}_2 \ar[r]^{\iota_n} & C^n \ar[r]^{\pi_n} & C_1^n \ar[r] &
  0}.
\]
\end{prop}

\begin{cor}\label{prop:ConeExactSequence}
  If $\func{f}{C_1}{C_2}$ is a chain map, then
\[
\xymatrix{\cdots \ar[r]^(.35){f^*_{n-1}} & H^{n-1}(C_2) \ar[r]^(.42){\iota^*_n} & H^n\big(\Cone(f)\big)
\ar[r]^(.55){\pi^*_n} & H^n(C_1) \ar[r]^{f^*_n} & H^n(C_2)
\ar[r]^(.52){\iota^*_{n+1}} & \cdots}
\]
\noi is an exact sequence.
In this case, maps $f_n^*$ are the maps induced in homology
by $f$.
\end{cor}

\subsubsection{Chain complex codes}
\label{sec:chain-complex-codes}

Now, we can state the purpose of this section.

\begin{prop}\label{prop:ChainCodes}
  To any length 3 piece of chain complex
  $C:=\left(C^{i_0-1}\xrightarrow{\p}C^{i_0}\xrightarrow{\p}C^{i_0+1}\right)$
  given with a basis $\BB$, one can associate
  a CSS code $C_C$ with parameter $\llbracket n;k;d\rrbracket$ where
  $n=\dim(C^{i_0})$, $k=\dim\big(H^{i_0} (C)\big)\bigg(=\dim\big(H_{i_0} (C)\big)\bigg)$
  and $d=\min\Big\{|x|_\BB\ \big|\ [x]\in H^{i_0} (C)\sqcup H_{i_0} (C),[x]\neq
  0\Big\}$, where $|\ .\ |_\BB$ denotes the $\BB$--weight, that is the number of non trivial
  coordinates in the basis $\BB$.
\end{prop}

\begin{proof}
  We set $\mathbf{H}_X:=\Mat_\BB(\p_{|C^{i_0}})$ and
  $\mathbf{H}_Z:=\Mat_\BB(\p_{|C^{i_0-1}})^t$.
  Since $\p^2=0$, we have that $\mathbf{H}_X\mathbf{H}_Z^t=0$ and the matrices $\mathbf{H}_X$ and
  $\mathbf{H}_Z$ define a CSS code $C_C$.
  Its length is trivially $\dim(C^{i_0})$.
  Its dimension is

  \begin{eqnarray*}
n-\rk(\mathbf{H}_X)-\rk(\mathbf{H}_Z) & = &
\dim(C^{i_0})-\rk(\p_{|C^{i_0}})-\rk(\p_{|C^{i_0-1}})\\
&=&\dim\big(\Ker(\p_{|C^{i_0}})\big)
-\rk(\p_{|C^{i_0-1}})\\
&=&\dim\left(\fractt{\Ker(\p_{|C^{i_0}})}/{\Im(\p_{|C^{i_0-1}})}\right)\
=\ \dim\big(H^{i_0}
(C)\big).
\end{eqnarray*}

To compute the minimum distance, we consider an error $E$ which
commutes with every element of $G$ but which is not in $G$.

If $E$ only involves $Z$ alterations, then it can be described by a
vector $v_E\in\F_2^n$ and the weight of $E$ is exactly
$|v_E|_\BB$. Since $E$ commutes with all the generators
of $G$ induced by the rows of $\mathbf{H}_X$, the vector $v_E$ is
orthogonal to all these rows and $v_E\in\Ker
(\p_{|C^{i_0}})$.
But $E\notin G$, so $v_E$ is not spanned by rows of $\mathbf{H}_Z$ and $v_E\notin
\Im(\p_{|C^{i_0-1}})$.
It follows that $E$ is non detectible iff $[v_E]$ is non zero in $H^{i_0}
(C)$.

If $E$ only involves $X$ alterations, then a similar reasoning at the
dual level shows that $E$ is non detectible iff $[v_E]$ is non zero in $H_{i_0}
(C)$.

Now, for a general $E$, we
factorize it as a product $E_XE_Z$ where $E_\alpha$ only involves
$\alpha$ alterations. Since every given generator of $G$ involves
only $X$ alterations or only $Z$ ones, the fact that $E$ commutes with
them implies that $E_X$ and $E_Z$ do.
But $E\notin G$, so at least one of $E_X$
 or $E_Z$ is not in $G$.
 We conclude by noting that the weight of $E$ is greater than each of the
 weights of $E_X$ and $E_Z$.
\end{proof}




\section{Khovanov homology}
\label{sec:khovanov-homology}

For more details on knot theory, the reader can refer to
\cite{Lickorish} or \cite{Kauffman}.
For details on Khovanov homology, the author advises Khovanov's
seminal paper \cite{Khovanov} for the general definition, \cite{KhovanovP} for the reduced case, Viro's elementary reformulation
\cite{Viro} and Shumakovich's survey \cite{Shumakovitch}. 

\subsection{Link diagrams}
\label{sec:basics-knot-theory}

A link is an embedding of a disjoint union of circles in $\R^3$
considered up to ambiant isotopies\footnote{Two maps
  $\func{f,g}{X}{Y}$ are said ambiant isotopic in $Y$ if there exists
  a continuous path of homeomorphisms $\func{\phi_t}{Y}{Y}$ such that
  $\phi_0=\Id_Y$ and $g=\phi_1\circ f$.} in $\R^3$.

The notion can be turned combinatorial by considering link diagrams.
They are generic projections, \ie with regular points and a finite number of transverse double points, of links into the plane $\R^2\times\{0\}$
together with an over/underpassing information for the strands at each
double point.

\begin{theo}[\cite{Reidemeister}]
  Every link admits diagrams and two given diagrams describe the same
  link iff they can be connected by ambiant isotopies in $\R^2$ and a finite number of the following
  Reidemeister moves\footnote{Two diagrams are connected by a
    Reidemeister move if they are identical outside a disk inside which
  they respectively correspond to the given pictures.}:
\[
\begin{array}{ccccc}
\dessin{1.2cm}{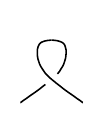} \sim \dessin{1.2cm}{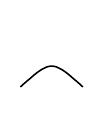} \sim \dessin{1.2cm}{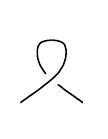}
&\hspace{.5cm}& \dessin{1.2cm}{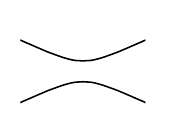} \sim \dessin{1.2cm}{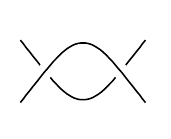} &\hspace{.5cm}&
\dessin{1.2cm}{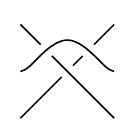} \sim \dessin{1.2cm}{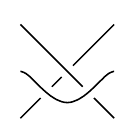}\\
  \textnormal{R1}^+\textnormal{ \& R1}^- && \textnormal{R2} && \textnormal{R3}
\end{array}.
\]
\end{theo}

A double point with over/underpassing information is called \emph{a
  crossing}.
There are two canonical ways to smooth (or resolve) a crossing:
\[
\vcenter{\hbox{
\xymatrix@!0@R=.8cm@C=3.8cm{
& \dessin{1cm}{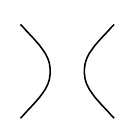}\left(\textrm{ or }\dessin{1cm}{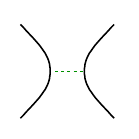}\right):\ 0\textrm{--resolution}\\
\dessin{1cm}{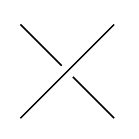} \ar[];[ur]!L \ar[];[dr]!L &\\
& \dessin{1cm}{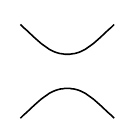}\left(\textrm{ or }\dessin{1cm}{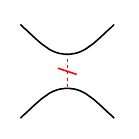}\right):\ 1\textrm{--resolution}
}}}.
\]
\noi The second pictures aim at keeping tracks of the resolved
crossing.
If $D$ is a link diagram, we call \emph{resolution of $D$} any map $\func{\phi}{\{\textrm{crossings of
  }D\}}{\{0,1\}}$, or equivalently the diagram $D_\phi$ obtained from $D$
by $\phi(c)$--resolving every crossing $c$ of $D$.
Resolution diagrams are \underline{not} considered up to isotopies and
different maps $\phi$ always lead to different resolution diagrams $D_\phi$.
Note that $D_\phi$ is a union of disjoint circles embedded
in the plane.
An \emph{enhanced resolution} $D_\phi^\sigma$ of $D$ is a resolution $D_\phi$ of $D$ together
with a labelling map $\func{\sigma}{\{\textrm{circles of
  }D_\phi\}}{\{1,X\}}$.
The labels can be seen as elements of $\fract{\F_2[X]\ }/{\ X^2}$, and later,
when dealing with combinations of enhanced diagrams, we will assume
multi-linearity for the labels.
Note that this $X$ is not related in any sense to the eponym Pauli error,
and actually, this notation will be dropped out by the end
of the section.

\begin{figure}
\[
  \begin{array}{ccccc}
 \dessin{2.5cm}{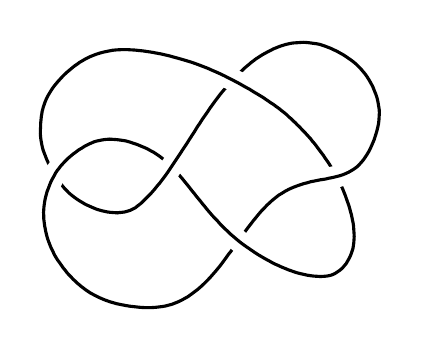} & \longmapsto & \dessin{2.5cm}{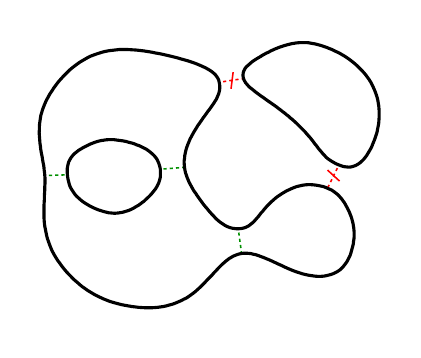} &
 \longmapsto & \dessin{2.5cm}{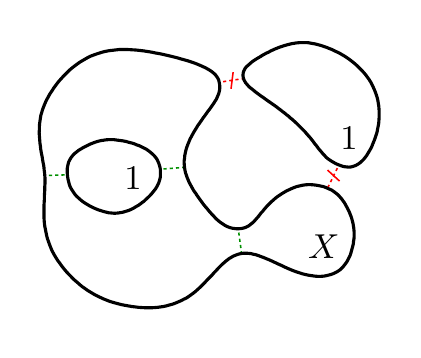}\\
\textrm{a diagram} && \textrm{a resolution} && \textrm{an enhanced resolution}
  \end{array}
\] 
 \caption{From diagrams to enhanced resolutions}
  \label{fig:FromDiagToRes}
\end{figure}

\subsection{Khovanov chain complex}
\label{sec:khov-chain-complex}

To any diagram $D$ with $n\in\N$ crossings, Khovanov theory associates
a length $n+1$ chain complex
\[
C(D):=\xymatrix{0\ar[r]&C^0\ar[r]^{\p_D}&C^1\ar[r]^{\p_D}&\cdots\ar[r]^{\p_D}&C^n\ar[r]&0}
\]
\noi defined as follows.
For $i\in\llbracket0,n\rrbracket$, $C^i$ is spanned over $\F_2$ by
enhanced resolutions of $D$ with exactly $i$ $1$--resolved crossings.
The map $\p_D$ is the $\F_2$--linear map defined
for a generator $D_\phi^\sigma$ by
\[
\p_D(D_\phi^\sigma)=\disp{\sum_{c\in\phi^{-1}(0)}}\p_c(D_\phi^\sigma)
\]
\noi where $\p_c(D_\phi^\sigma)$ is a sum of enhanced resolutions over
$D_{\phi+\delta_c}$, with $\delta_c$ the Kronecker delta.
The resolution $D_{\phi+\delta_c}$ is nothing but the resolution obtained by changing the smoothing of $c$.
Before stating the enhancing rules, let us note that
$D_{\phi+\delta_c}$ differs from $D_\phi$ by the merging of two
circles into one or the splitting of a circle into two.
Now, the rules are:
\begin{itemize}
\item[-] the untouched circles keep their labels unchanged;
\item[-] if two circles are merging, then the resulting circle is labelled by the product of
the labels in $\fract{\F_2[X]\ }/{\ X^2}$. Note that a $0$--label just means no
contribution;
\item[-] if one $1$--labelled circle is splitting, then there are two
contributions obtained as the two ways to distribute $1$ and $X$ to the
two new circles;
\item[-] if one $X$--labelled circle is splitting, then there is only one
contribution obtained by labelling both new circles by $X$.
\end{itemize}
These rules are summarized in Fig. \ref{fig:Rules1}.

\begin{figure}
  \[
  \begin{array}{ccc}
    \begin{array}{ccc}
     \dessin{1.2cm}{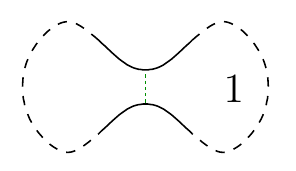} & \xmapsto{\ \p_c\ } &
     \dessin{1.2cm}{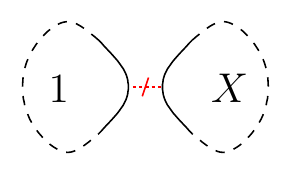}+\dessin{1.2cm}{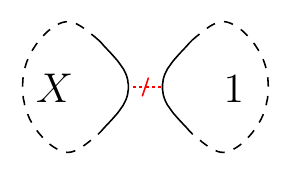}\\[.5cm]
     \dessin{1.2cm}{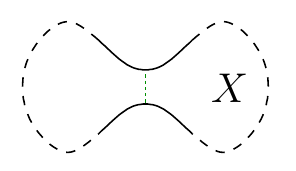} & \xmapsto{\ \p_c\ } &  \dessin{1.2cm}{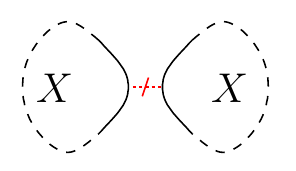}
    \end{array}
    & \hspace{.5cm} &
     \begin{array}{ccc}
     \dessin{1.2cm}{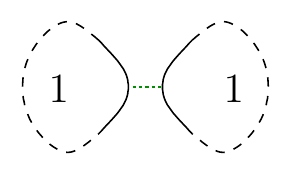} & \xmapsto{\ \p_c\ } & \dessin{1.2cm}{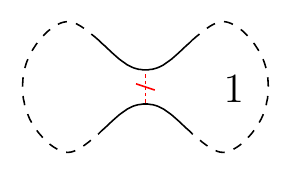}\\[.5cm]
          \dessin{1.2cm}{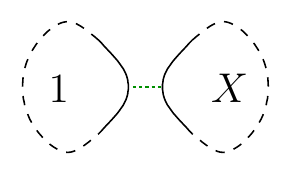} & \xmapsto{\ \p_c\ } &
          \dessin{1.2cm}{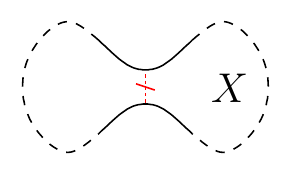}\\[.5cm]
     \dessin{1.2cm}{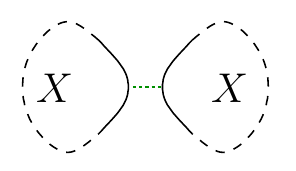} & \xmapsto{\ \p_c\ } & \textrm{zero}
    \end{array}
  \end{array}
\]
  \caption{Rules for labelling in the differential: {\scriptsize here
      (and throughout the paper),
      only the modified part is depicted, the rest of the resolutions being
    identical on both sides of the arrows}}
  \label{fig:Rules1}
\end{figure}

\begin{prop}[\cite{Khovanov} Prop. 8, \cite{Viro} Th. 5.3.A]
  The map $\p_D$ satisfies $\p_D\circ\p_D=0$.
\end{prop}

\begin{remarques}$\ $
  \begin{enumerate}
  \item The construction was originally given with $\Z$--coefficients instead
    of $\F_2$--ones. It can therefore be adapted to any ring.
  \item Khovanov homology is usually defined with a second grading $j$
    on $C(D)$, namely
    $j(D_\phi^\sigma)=|\sigma^{-1}(X)|-|\sigma^{-1}(1)|-|\phi^{-1}(1)|$
    where $|\ .\ |$ stands for cardinality.
    Since the differential $\p_D$ respects this grading $j$, the chain
    complex $C(D)$ splits into several chain complexes, one for each
    value of $j$.
    However, this grading is not relevant for the purpose of the present paper.
  \end{enumerate}
\end{remarques}

\subsection{Change of variable}
\label{sec:change-variable}

With this basis, Khovanov complexes are not really efficient for quantum
codes since non trivial homology elements can easily have small
weight.
To change this matter of fact, we consider another set of generators, where
labels are not anymore $1$ and $X$ but signs $-:=1$ and $+:=1+X$. A
label $+$ for a circle means the sum of the two generators for
which the circle is labelled by 1 or by $X$, all the others circles
being identically labelled.
The differential is then kind of symmetrized as pointed in Fig. \ref{fig:Rules2}.

\begin{figure}
  \[
  \begin{array}{ccc}
    \begin{array}{ccc}
     \dessin{1.2cm}{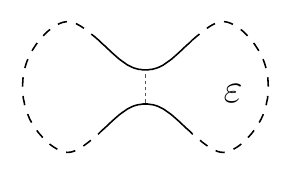} & \xmapsto{\ \p_c\ } & \disp{\sum_{\eta}}\dessin{1.2cm}{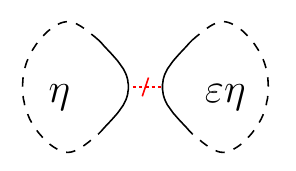}
    \end{array}
    & \hspace{.5cm} &
     \begin{array}{ccc}
     \dessin{1.2cm}{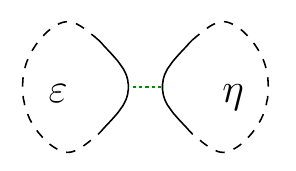} & \xmapsto{\ \p_c\ } & \dessin{1.2cm}{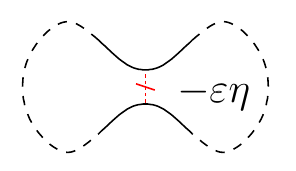}
    \end{array}
  \end{array}
\]
  \caption{Modified rules for labelling in the differential:
    {\scriptsize here, $\e$ and $\eta$ are element of $\{-,+\}$ and the product
      is the obvious one}}
  \label{fig:Rules2}
\end{figure}

\begin{remarque}
  The new set of generators is not anymore graded with regard to the
  second grading $j$. That is essentially why $j$ is not relevant here.
\end{remarque}

\subsection{Reidemeister moves invariance}
\label{sec:reid-moves-invar}

The Khovanov complex $C(D)$ depends heavily on the diagram
$D$, but if considering the homology, then $\Kh(D):=H^*\big(C(D)\big)$ depends essentially on the underlying
link.
Indeed, the following theorem makes explicit the behavior of $\Kh(D)$
under Reidemeister moves.

\begin{figure}
\vspace{1.5cm}
  \[
\hspace{-1.5cm}
\begin{array}{ccc}
  \fbox{
    $ \begin{array}{ccc}
    C\Big(\dessin{1.5cm}{R1}\Big) & \longrightarrow &
    C\Big(\dessin{1.5cm}{R1+}\Big)\\[-.2cm]
    \dessin{1.2cm}{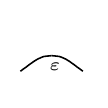} & \longmapsto &
    \dessin{1.2cm}{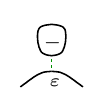}+\dessin{1.2cm}{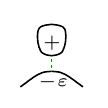}\\[1.5cm]
    C\Big(\dessin{1.5cm}{R1+}\Big) & \longrightarrow &
    C\Big(\dessin{1.5cm}{R1}\Big)\\[-.2cm]
    \dessin{1.2cm}{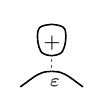} & \longmapsto &
    \dessin{1.2cm}{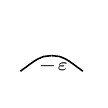}\\
    \textrm{others} & \longmapsto & \textrm{zero}
    \end{array}$
  }

&\hspace{.5cm}&

  \fbox{
    $\begin{array}{ccc}
    C\Big(\dessin{1.5cm}{R1}\Big) & \longrightarrow &
    C\Big(\dessin{1.5cm}{R1-}\Big)\\[-.2cm]
    \dessin{1.2cm}{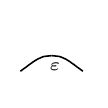} & \longmapsto &
    \dessin{1.2cm}{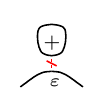}\\[1.5cm]
    C\Big(\dessin{1.5cm}{R1-}\Big) & \longrightarrow &
    C\Big(\dessin{1.5cm}{R1}\Big)\\[-.2cm]
    \dessin{1.2cm}{R1-_2} & \longmapsto &
    \dessin{1.2cm}{R1-_1}\\
    \textrm{others} & \longmapsto & \textrm{zero}
    \end{array}$
}\\
\textnormal{Reidemeister move R1}^+
&&
\textnormal{Reidemeister move R1}^-\\[1cm]

  \fbox{
    $\begin{array}{ccc}
    C\Big(\dessin{1.5cm}{R2}\Big) & \longrightarrow &
    C\Big(\dessin{1.5cm}{R2+}\Big)\\[-.2cm]
    \dessin{1.2cm}{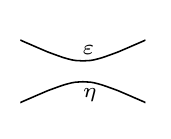} & \longmapsto &
    \dessin{1.2cm}{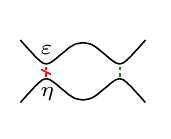}+\dessin{1.2cm}{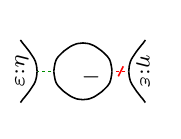}\\[1.5cm]
    C\Big(\dessin{1.5cm}{R2+}\Big) & \longrightarrow &
    C\Big(\dessin{1.5cm}{R2}\Big)\\
    \dessin{1.2cm}{R2_2} & \longmapsto &
    \dessin{1.2cm}{R2_1}\\
    \dessin{1.2cm}{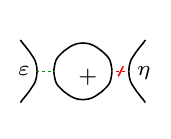} & \longmapsto &
    \dessin{1.2cm}{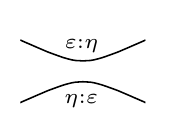}\\
    \textrm{others} & \longmapsto & \textrm{zero}
  \end{array}$
}
&&
  \fbox{
    $\begin{array}{c}
\begin{array}{ccc}
    C\Big(\dessin{1.5cm}{R3--}\Big) & \longrightarrow &
    C\Big(\dessin{1.5cm}{R3++}\Big)\\[-.2cm]
    \dessin{1.2cm}{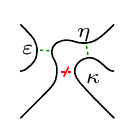} & \longmapsto &
    \dessin{1.2cm}{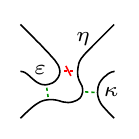}+\dessin{1.2cm}{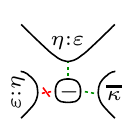}\\ 
    \dessin{1.2cm}{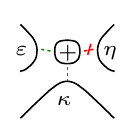} & \longmapsto &
    \dessin{1.2cm}{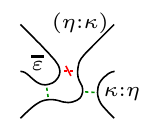}+\dessin{1.2cm}{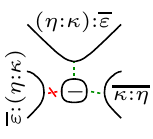}+\dessin{1.2cm}{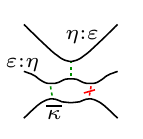}\\
    \dessin{1.2cm}{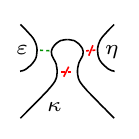} & \longmapsto &
    \dessin{1.2cm}{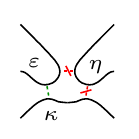}\\   
    \dessin{1.2cm}{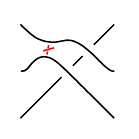} & \longmapsto &
    \dessin{1.2cm}{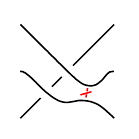}\\
    \textrm{others} & \longmapsto & \textrm{zero}
  \end{array}
    \end{array}$
}\\
\textnormal{Reidemeister move R2}
&&
\textnormal{Reidemeister move R3}
\end{array}
\]
  \caption{Invariance chain maps:  {\scriptsize only
      the part involved in the Reidemeister move is depicted, the rest
    of the diagrams are identical on each side; $\varepsilon$:$\eta$ and
      $\eta$:$\varepsilon$ are the two labels (maybe a sum of)
      obtained when merging/splitting circles with labels
      $\varepsilon$ and $\eta$; overlining a label means
    that it may be modified if, outside the depicted part, its circle is connected
    to the splitting/merging ones; a unresolved crossing
      stands for any of its resolutions, the map is then the natural
      one-to-one one}}
  \label{fig:Invariance}
\end{figure}

\begin{theo}[\cite{Viro} sections 5.6 \& 5.7, \cite{Ito}]\label{prop:Invariance}
  Let $D_1$ and $D_2$ be two link diagrams connected by a Reidemeister
  move (with $D_2$ having greater or equal number of crossings than
  $D_1$).
  Then the chain maps given in Fig. \ref{fig:Invariance} induce
  isomorphisms between $\Kh(D_2)$ and $\Kh(D_1)\{\eta\}$ where
  $\{\ .\ \}$ denotes a shift in the grading and $\eta=1$ if the Reidemeister move
  is $\textnormal{R1}^-$ or $\textnormal{R2}$ and $\eta=0$ otherwise.
\end{theo}

\begin{remarque}
  There is a canonical way to shift Khovanov homology so it becomes
  really invariant under Reidemeister moves (\cite{Khovanov}), but this is not relevant
  for our purpose.
\end{remarque}

\subsection{Basic properties}
\label{sec:basic-properties}

Khovanov homology does behave quite nicely under certain usual operations
on knots.

\begin{prop}[\cite{Khovanov} Cor. 12]
  If $D_1$ and $D_2$ are two link diagrams, then $C(D_1\sqcup D_2)\cong
  C(D_1)\otimes C(D_2)$ so $\Kh(D_1\sqcup D_2)\cong
  \Kh(D_1)\otimes \Kh(D_2)$.
\end{prop}

\begin{prop}[\cite{Khovanov} Prop. 32]\label{prop:Duality}
  For any link $D$ with $n$ crossings and for every $i\in\llbracket
  0,n\rrbracket$, $\Kh^i(D!)\cong \Kh^\vee_{n-i}(D)$ where $D!$ is the mirror
  image of $D$, \ie the link obtained by swapping the under and the
  over strands at every crossings, and $\vee$ stands for duality.
  Besides, the isomorphism is induced by the generator-to-generator chain map
  $\func{m}{C^\vee_{n-i}(D)}{C^i(D!)}$ defined by
  $m({D^\sigma_\phi}^\vee)=D!_{1-\phi}^{-\sigma}$.
\end{prop}
\begin{remarque}
  This analogue of Poincaré duality is of special interest since it enables to deal with
  dual chain complexes while staying in the frame of Khovanov complexes.
\end{remarque}

\subsection{Reduced Khovanov homology}
\label{sec:reduc-khov-homology}

There is a reduced Khovanov homology defined for pointed link, \ie links
with a marked point on it.
The definition is nearly the same except the marked point induces a pointed circle in
every resolution, and we force it to be labelled by $X$, that is the
sum of labels $-$ and $+$.
It leads to the additional labelling rules for the differential given in Fig. \ref{fig:Rules3}.

\begin{figure}
  \[
  \begin{array}{ccc}
    \begin{array}{ccc}
     \dessin{1.2cm}{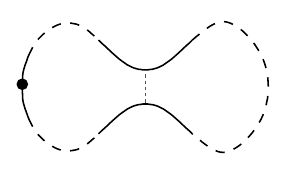} & \xmapsto{\ \p_c\ } &
     \disp{\sum_\e}\dessin{1.2cm}{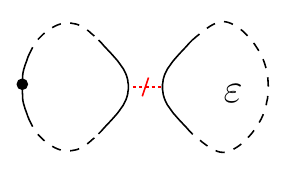}
    \end{array}
    & \hspace{.5cm} &
     \begin{array}{ccc}
     \dessin{1.2cm}{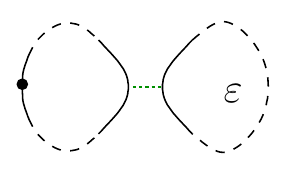} & \xmapsto{\ \p_c\ } & \dessin{1.2cm}{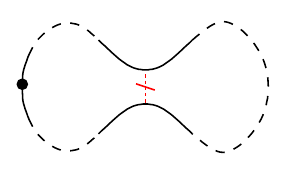}
    \end{array}
  \end{array}
\]
  \caption{Extra rules for labelling in the reduced differential:
    {\scriptsize here $\e$ is an element of $\{-,+\}$}}
  \label{fig:Rules3}
\end{figure}
\begin{prop}[\cite{Shumakovitch} Theo. 2.6]
  If $D_\bullet$ is a pointed version of a link diagram $D$, then
  $\Kh(D)\cong \Kh(D_\bullet)\oplus\Kh(D_\bullet)$.
\end{prop}
\begin{prop}\label{prop:ConnectedSum}
  If $D_1$ and $D_2$ are two pointed link diagrams, then $C(D_1\# D_2)\cong
  C(D_1)\otimes C(D_2)$ so $\Kh(D_1\# D_2)\cong
  \Kh(D_1)\otimes \Kh(D_2)$, where $\#$ is the connected sum operation
  done on the two marked points (see Fig. \ref{fig:ConnectedSum}).
\end{prop}

\begin{figure}
  \[
\dessin{1.8cm}{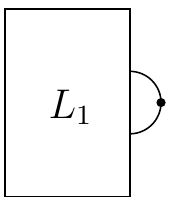}\ \#\ \dessin{1.8cm}{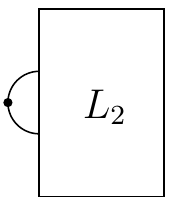}=\dessin{1.8cm}{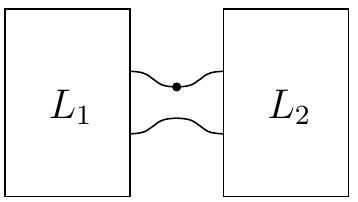}
\]
  \caption{Connected sum for pointed links}
  \label{fig:ConnectedSum}
\end{figure}

\subsection{Exact sequence}
\label{sec:exact-sequence}

Let $D$ be a link diagram (possibly pointed) and $c$ a crossing of $D$.
We denote by $D_0$ and $D_1$ the diagrams obtained, respectively, by
$0$--resolving and $1$--resolving $c$.
It follows from the definition that:
\begin{prop}\label{prop:KhovCone}
  $C(D)\cong \Cone\big(\func{\p_c}{C(D_0)}{C(D_1)}\big)$.
\end{prop}

If denoting by $\func{\alpha}{C(D_1)}{C(D)}$ and
$\func{\beta}{C(D)}{C(D_2)}$ the natural injection and surjection,
then Prop. \ref{prop:ConeExactSequence} implies:

\begin{cor}[\cite{Viro} section 6.2]
The long sequence
  \[
  \xymatrix{\cdots\ar[r]^(.35){\p^*_c}&\Kh^{i-1}(D_1)\ar[r]^(.55){\alpha^*}&\Kh^i(D)\ar[r]^{\beta^*}&\Kh^i(D_0)\ar[r]^(.45){\p_c^*}&\Kh^i(D_1)\ar[r]^(.6){\alpha^*}&\cdots}.
  \]
\noi is exact.
\end{cor}

\subsection{Weight considerations}
\label{sec:weight-cons}
As far as the author knows, weight of representives for non-zero elements in Khovanov
homology have not been studied yet.
This section aims at presenting some first thoughts toward this direction.

For every chain complex $C:=\oplus_{i\in\Z}C^i$ and every integer
$i\in\Z$, we denote by
$d^i_C:=\min\big\{|x|\ \big|\ x\in C^i, 0[x]\in H^i(C)\setminus\{0\}\big\}$.
In the case of Khovanov homology, we will write, for a diagram $D$ and
an integer $i\in\N$, $d^i_D$ for $d^i_{C(D)}$.
\begin{prop}\label{prop:WeightMap}
  Let $C_1$ and $C_2$ be two chain complexes.
  If a chain map (which possibly shifts the homological grading) $\func{\psi}{C^i_1}{C^j_2}$, with $i,j\in\N$, induces an
  injective map in homology, then $kd^i_{C_1}\geq d^j_{C_2}$ where
  $k:=\max\big\{|\psi(x)|\ \big|\ x\textrm{ generator of
  }C^i_1\big\}$.\\
  Moreover, if $k=1$, if the map $\psi$ is also injective and if a minimally
  weighted homology-surviving element of $C^j_2$ is on the image of
  $\psi$, then the
  inequality becomes an equality.
\end{prop}
\begin{proof}
  Let $x\in C^i_1$ be such that $[x]\neq 0$ and $|x|=d^i_{C_1}$.
  In one hand, we have $|\psi(x)|\leq k|x|$ but on the other hand,
  since $\psi^*$ is injective, $\psi^*\big([x]\big)=\big[\psi(x)\big]\neq 0$ so
  $|\psi(x)|\geq d^j_{C_2}$.

  Now, if all the conditions of the second part of the statement hold,
  we can find $y\in C^j_2$ and $x\in C^i_1$ such that $[y]\neq
  0$, $|y|$ is minimal and $\psi(x)=y$.
  Then $\psi\big(\p_{C_1}(x)\big)=\p_{C_2}(y)=0$, but $\psi$ is
  injective so $\p_{D_1}(x)=0$ and since $\psi^*\big([x]\big)=[y]\neq0$, $[x]\neq 0$.
  But $\psi$ is injective and $k=1$ so $|y|=\big|\psi(x)\big|=|x|$.
  It follows that $d^i_{C_1}\leq d^j_{C_2}$ and hence $d^i_{C_1}=d^j_{C_2}$.
\end{proof}

\begin{cor}\label{cor:ReidemeisterWeight}
  With obvious notation for diagrams differing from Reidemeister moves, we have
  for any $i\in\N$
  \[
  \begin{array}{ccc}
    d^i_{\dessin{.3cm}{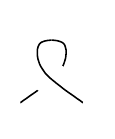}}=2d^i_{\dessin{.3cm}{R1}} & \hspace{1cm} &
    d^{i+1}_{\dessin{.3cm}{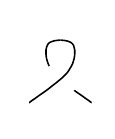}}=d^i_{\dessin{.3cm}{R1}}\\[.5cm]
    \frac{1}{3}d^i_{\dessin{.3cm}{R2}}\leq
    d^{i+1}_{\dessin{.3cm}{R2+}}\leq2d^i_{\dessin{.3cm}{R2}} &&
    \frac{1}{8}d^i_{\dessin{.3cm}{R3++}}\leq
    d^i_{\dessin{.3cm}{R3--}}\leq8d^i_{\dessin{.3cm}{R3++}}
  \end{array}.
  \]
\end{cor}
\begin{proof}
  Most of the statement is a direct application of
  Prop. \ref{prop:Invariance} and \ref{prop:WeightMap}.
  Only $d^i_{\dessin{.3cm}{R1+b}}\geq 2d^i_{\dessin{.3cm}{R1}}$ needs a
  further argumentation.
  Let $x\in C^i(\dessin{.8cm}{R1+})$ be a representative of a non-zero
  element of the homology.
  We can decompose it as $x=a_++a_-+b$ with $a_+$ (resp. $a_-$) a sum of generators of the form $\dessin{.8cm}{R1+_4}$
  (resp. $\dessin{.8cm}{R1+_2}$) and $b$ a sum of generators
  of the form $\dessin{.8cm}{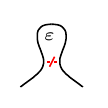}$.
  Since $x$ represents an element of the homology, we know that
  $\p_{\dessin{.3cm}{R1+b}}(x)=0$.
  Looking at the part which lies in resolutions of the form
  $\dessin{.8cm}{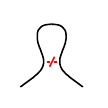}$, we obtain
  $A_++A_-+\p_{\dessin{.3cm}{R1+_7}}(b)=0$ where $A_-$ (resp. $A_+$) is
  an element of $\dessin{.7cm}{R1+_7}$
  obtained from $a_-$ (resp. $a_+$) by removing the ``$-$''--labelled
  circle and performing a small isotopy
  (resp. removing the ``$+$''--labelled
  circle, inverting the sign of $\varepsilon$ and performing a small isotopy).
  In particular $|A_+|=|a_+|$ and $|A_-|=|a_-|$.
  Applying backward the small isotopy, we obtain
  $\widetilde{A}_++\widetilde{A}_-+\p_{\dessin{.3cm}{R1}}(\widetilde{B})=0$
  in $C^i(\dessin{.7cm}{R1})$.
  We deduce that $[\widetilde{A}_+]=[\widetilde{A}_-]$ in
  $\Kh^i(\dessin{.7cm}{R1})$.
  But the image of $x$ under the $\textnormal{R1}^+$--chain
  quasi-isomorphism is precisely $\widetilde{A}_+$.
  So $[\widetilde{A}_+]=[\widetilde{A}_-]\neq 0$ and
  $|\widetilde{A}_+|,|\widetilde{A}_-|\geq d^i_{\dessin{.3cm}{R1}}$.
  Finally, $|x|\geq
  |a_+|+|a_-|=|\widetilde{A}_+|+|\widetilde{A}_-|\geq 2d^i_{\dessin{.3cm}{R1}}$.
\end{proof}
\begin{remarque}
  Computations and the fact that awkward generators are part of acyclic subcomplexes suggest that those naïve
  bounds are far from being sharp for Reidemeister moves R2 and R3. 
\end{remarque}
\begin{question}\label{quest:ConjReid}
  Do Reidemeister moves R2 always double minimal distances, and do
  Reidemeister moves R3 always preserve it ? If true, Khovanov
  homology would hide inner invariants on each degree supporting a
  non trivial homology.
\end{question}



\section{Unknot codes}
\label{sec:unknot-codes}

For every $\ell\in\N$, we consider the following diagram $\Duk{\ell}$ of the
pointed unknot with $2\ell$ crossings:

\[
\dessin{1cm}{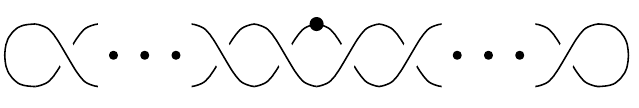}.
\]

We call \emph{$\ell^\textrm{th}$ unknot code} the code obtained from
$\big(C^{\ell-1}(\Duk{\ell})\xrightarrow{\p_{\Duk{\ell}}} C^\ell(\Duk{\ell})\xrightarrow{\p_{\Duk{\ell}}} C^{\ell+1}(\Duk{\ell})\big)$.
Its parameters are denoted by $\llbracket n_\ell;k_\ell;d_\ell\rrbracket$.

\subsection{Length}
\label{sec:length-unknot}

\begin{prop}
  $n_\ell\sim \frac{3^{2\ell+1}}{\sqrt{8\pi \ell}}$ as $\ell$ tends to infinity.
\end{prop}
\begin{proof}
  When 0--resolving all the crossings, we obtain
  $\dessin{.7cm}{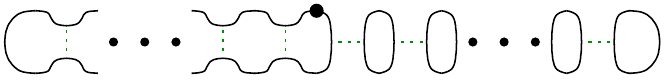}$ with $\ell$ undotted circles.
  Swapping the resolution of one of the $\ell$ crossings on the left
  creates a new undotted circle.
  On the contrary, swapping the resolution of one of the $\ell$ crossings
  on the right reduces by one the number of undotted circles.
  Now we gather the generators of $C^\ell(\Duk{\ell})$ according to the number
  of $1$--resolved crossings among the $\ell$ left ones.
  We obtain $n_\ell=\disp{\sum_{r=0}^\ell}{\ell\choose r}{\ell\choose
    \ell-r}2^{\ell+r-(\ell-r)}=\disp{\sum_{r=0}^\ell}\left[{\ell\choose
      r}2^r\right]^2$.
Then, using the formula of Prop. \ref{prop:SumCarre} for $x=2$, we get $n_\ell\sim
\frac{3^{2 \ell+1}}{\sqrt{8\pi \ell}}$.
\end{proof}

\subsection{Dimension and minimum distance}
\label{sec:dimens-minim-dist-unknot}

\begin{prop}
  $k_\ell=1$ and $d_\ell=2^\ell$.
\end{prop}
\begin{proof}
  To pass from $\Duk{\ell}$ to $\Duk{\ell+1}$, one can perform two R1 moves (one $\textrm{R1}^+$ and one $\textrm{R1}^-$).
  Now the statement on $k_\ell$ follows from Prop. \ref{prop:Invariance} and
  the statement on $d_\ell$ from Prop. \ref{cor:ReidemeisterWeight}.
\end{proof}

\subsection{Sparseness}
\label{sec:Sparseness-unknot}

\begin{prop}
  The weight of each row in the $\ell^\textrm{th}$ unknot code is
  $O\big(\ln(n_\ell)\big)$ as $\ell$ increases.
\end{prop}
\begin{proof}
  It is clear from Khovanov homology construction that each row has
  between $\ell+1$ and $2(\ell+1)$ non trivial entries.
  Since $8^\ell\leq n_\ell\leq 9^\ell$ for sufficiently large $\ell$, the result follows.
\end{proof}




\section{Unlink codes}
\label{sec:unlink-codes}

For every $\ell\in\N$, we consider the following diagram $\Dul{\ell}$ of the pointed $(\ell+1)$--unlink:
\[
\dessin{2.5cm}{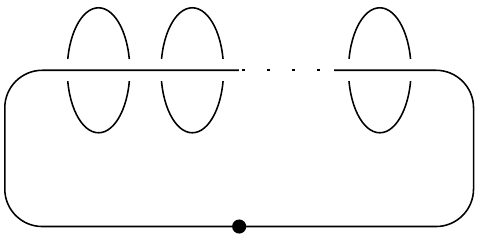}.
\]

We call \emph{$\ell^\textrm{th}$ unlink code} the code obtained from
$\big(C^{\ell-1}(\Dul{\ell})\xrightarrow{\p_\Dul{\ell}} C^\ell(\Dul{\ell})\xrightarrow{\p_\Dul{\ell}} C^{\ell+1}(\Dul{\ell})\big)$.
Its parameters are denoted by $\llbracket n_\ell;k_\ell;d_\ell\rrbracket$.

\subsection{Case $\ell=1$}
\label{sec:case-l=1}

It follows from Prop. \ref{prop:ConnectedSum} that
$C(\Dul{\ell})\cong C(\Dul{1})^{\otimes \ell}$.
It is hence worthwhile to deal the case $\ell=1$ in detail.

It can be directly computed that $C(\Dul{1})\cong C^\vee(\Dul{1})$ has six generators:
\[
\begin{array}{lc}
  \textrm{degree }0: & a:=\dessin{1.2cm}{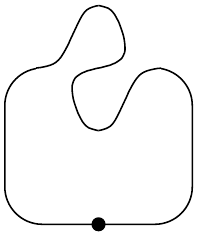} \\
   \textrm{degree }1: & b_1:=\dessin{1.2cm}{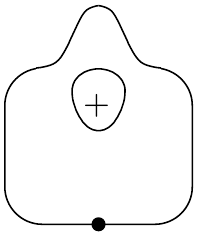},\ 
    b_2:=\dessin{1.2cm}{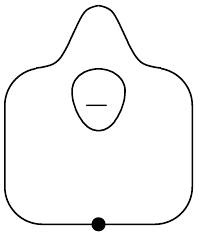},\ 
    b_3:=\dessin{1.2cm}{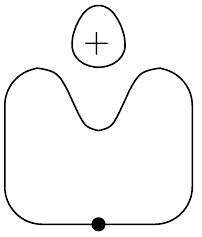},\ 
    b_4:=\dessin{1.2cm}{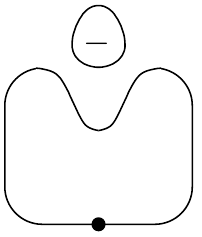}\\
  \textrm{degree }2: & c:=\dessin{1.2cm}{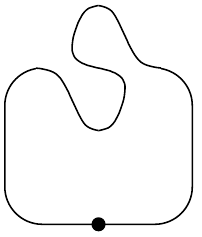} 
\end{array}
\]
The differential is $\p_{\Dul{1}}(a)= b_1+b_2+b_3+b_4$ and
$\p_{\Dul{1}}(b_1)=\p_{\Dul{1}}(b_2) =\p_{\Dul{1}}(b_3) =\p_{\Dul{1}}(b_4) =c$.
The non-zero elements of the homology are then represented by sums
$b_i+b_j$ with $i\neq j\in\llbracket1,4\rrbracket$ and two such sums
are equivalent iff their supports are disjoint.
The homology is then of rank 2 and its three non trivial elements are
$[b_1+b_2]=[b_3+b_4]$, $[b_1+b_3]=[b_2+b_4]$ and $[b_1+b_4]=[b_2+b_3]$.

\subsection{Length}
\label{sec:length-unlink}



\begin{prop}
  $n_\ell\sim \sqrt{\frac{3}{2\pi \ell}}6^\ell$ as $\ell$ tends to infinity.
\end{prop}
\begin{proof}
  Since Prop. \ref{prop:ConnectedSum}, we have
  $C(\Dul{\ell})=C(\Dul{1})^{\otimes \ell}$. It follows then that
  $\dim\big(C^\ell (\Dul{\ell})\big)$ is the coefficient of degree $\ell$
  in $(1+4t+t^2)^\ell$, that is the constant term in $(t^{-1}+4+t)^\ell$.
  But
\[
(t^{-1}+4+t)^\ell=\big((t^{-\frac{1}{2}}+t^{\frac{1}{2}})^2+2\big)^\ell=\sum_{r=0}^\ell{\ell\choose
    r}(t^{-\frac{1}{2}}+t^{\frac{1}{2}})^{2r}2^{\ell-r}=2^\ell\sum_{r=0}^\ell\frac{{\ell\choose
    r}}{2^r}\sum_{l=0}^{2r}{2r\choose l}t^{r-l}
\]
\noi so $n_\ell=2^\ell\disp{\sum_{r=0}^\ell}\frac{{\ell\choose r}{2r\choose
    r}}{2^r}$. Then we use Prop. \ref{prop:Tn} to conclude.
\end{proof}

\subsection{Dimension}
\label{sec:dimension-unlink}

\begin{prop}
  $k_\ell=2^\ell$.
\end{prop}

This is a direct consequence of Prop. \ref{prop:ConnectedSum}.

\subsection{Minimum distance}
\label{sec:minimum-distance-unlink}

\begin{prop}
  $d_\ell=2^\ell$.
\end{prop}
\begin{proof}
  It is easily seen that there is a differential-preserving one-to-one correspondance between
  generators of $C(\Dul{\ell})$ and $C(\Dul{\ell}!)\cong C^\vee(\Dul{\ell})$.
  It is hence sufficent to deal with $C(\Dul{\ell})$.

  By induction on $\ell$, we prove a sligthly stronger result: $2^\ell$ is
  the minimum distance and it is reached for any non trivial element
  of the homology.
  This is trivial for $\ell=0$ (and it has been checked for $\ell=1$).
  Now, we assume the assertion is true for a given $\ell\in\N$.

  Since $C(\Dul{\ell+1})\cong C(\Dul{\ell})\otimes C(\Dul{1})$, any element $A$ of
  $C^k(\Dul{\ell+1})$, for $k\in\llbracket 0,2 \ell+2\rrbracket$ can be
  decomposed into the following form
\[
A=
\left\{
\begin{array}{cl}
  x\otimes a & \in C^k(\Dul{\ell})\otimes C ^0(\Dul{1})\\
  + &\\
  \sum_{i=1}^4y_i\otimes b_i & \in
  C^{k-1}(\Dul{\ell})\otimes C^1(\Dul{1})\\
  + &\\
  z\otimes c & \in C^{k-2}(\Dul{\ell})\otimes C^2(\Dul{1})
\end{array}
\right..
\]
  Thus, we have
\[
\p_{\Dul{\ell+1}}(A)=
\left\{
\begin{array}{cl}
  \p_{\Dul{\ell}}(x)\otimes a & \in C^{k+1}(\Dul{\ell})\otimes C^0(\Dul{1})\\
  + &\\
  \sum_{i=1}^4\big(x+\p_{\Dul{\ell}}(y_i)\big)\otimes b_i& \in
  C^k(\Dul{\ell})\otimes C^1(\Dul{1})\\
  + &\\
  \big(y_1+y_2+y_3+y_4+\p_{\Dul{\ell}}(z)\big)\otimes c & \in C^{k-1}(\Dul{\ell})\otimes
  C^2(\Dul{1})
\end{array}
\right..
\]
\begin{lemme}
  If $\big([w_1],\cdots,[w_{2^\ell}]\big)$ is a basis for $\Kh(\Dul{\ell})$, then
  \[
\Big(\big[w_1\otimes(b_1+b_2)\big],\cdots,\big[w_{2^\ell}\otimes(b_1+b_2)\big],\big[w_1\otimes(b_1+b_3)\big],\cdots,\big[w_{2^\ell}\otimes(b_1+b_3)\big]\Big)
\]
is a basis for $\Kh(\Dul{\ell+1})$.
\end{lemme}
\begin{proof}
  Elements of the form $w_i\otimes(b_i+b_j)$, for
  $i,j\in\llbracket1,4\rrbracket$ are clearly in the kernel of
  $\p_{\Dul{\ell+1}}$.
  If
\[
\sum_{i=1}^{2^\ell}\alpha_i \big[w_i\otimes(b_1+b_2) \big]+\beta_i
  \big[w_i\otimes(b_1+b_3)\big]=0
\]
\noi with
  $(\alpha_i),(\beta_i)\in\F_2^{2^\ell}$, then there exists
  $A\in C_\ell (\Dul{\ell+1})$ such that
\[
\sum_{i=1}^{2^\ell}\alpha_i w_i\otimes(b_1+b_2)+\beta_i
  w_i\otimes(b_1+b_3)=\p_{\Dul{\ell+1}}(A)
\]
\noi and hence, with the notation above, and by looking at the $\ .\
\otimes b_i$ parts,
\[
\left\{
\begin{array}{l}
  x+\p_{\Dul{n}}(y_1)=\sum_{i=1}^{2^\ell} (\alpha_i+\beta_i)w_i\\[.1cm]
  x+\p_{\Dul{n}}(y_2)=\sum_{i=1}^{2^\ell} \alpha_iw_i\\[.1cm]
  x+\p_{\Dul{n}}(y_3)=\sum_{i=1}^{2^\ell} \beta_iw_i\\[.1cm]
  x+\p_{\Dul{n}}(y_4)=0
\end{array}
\right..
\]
It follows that $\sum_{i=1}^{2^\ell} \alpha_iw_i=\p_{\Dul{\ell}}(y_2+y_4)$ and
$\sum_{i=1}^{2^\ell} \beta_iw_i=\p_{\Dul{\ell}}(y_3+y_4)$.
This means that
$\sum_{i=1}^{2^\ell} \alpha_i[w_i]=\sum_{i=1}^{2^\ell} \beta_i[w_i]=0$ and
hence that $\alpha_i=\beta_i=0$ for every $i\in\llbracket1,2^\ell\rrbracket$. 
\end{proof}
\begin{lemme}
  If $[A]$ is a non trivial element of $\Kh(\Dul{\ell+1})$ then $|A|\geq2^{\ell+1}$.
\end{lemme}
\begin{proof}
  If $A=\alpha\otimes a+\sum_{i=1}^4 \beta_i\otimes b_i+\gamma\otimes
  c$, then $|A|=|\alpha|+\sum_{i=1}^4|\beta_i|+|\gamma|$.
  According the precedent lemma, there exists
  $(v,w)\in\Ker(\p_{\Dul{\ell}})$ such that $[A]=\big[(v+w)\otimes b_1+v\otimes
  b_2+w\otimes b_3\big]$ with $([v],[w])\neq(0,0)$.
  As above, it follows, that
  \[
  \left\{
  \begin{array}{l}
    x+\p_{\Dul{\ell}}(y_1)=\beta_1+v+w\\[.1cm]
    x+\p_{\Dul{\ell}}(y_2)=\beta_2+v\\[.1cm]
    x+\p_{\Dul{\ell}}(y_3)=\beta_3+w\\[.1cm]
    x+\p_{\Dul{\ell}}(y_4)=\beta_4.
  \end{array}
  \right.
  \]
  If $[v]\neq0$, then $\beta_1+\beta_3=v+\p_{\Dul{\ell}}(y_1+y_3)$ and $\beta_2+\beta_4=v+\p_{\Dul{\ell}}(y_2+y_4)$ so
  $[\beta_1+\beta_3]=[\beta_2+\beta_4]$ is a non trivial element of
  $\Kh(\Dul{\ell})$ so $|\beta_1+\beta_3|\geq2^\ell$ and
  $|\beta_2+\beta_4|\geq2^\ell$.
  Finally
\[
|A|\geq |\beta_1|+|\beta_2|+|\beta_3|+|\beta_4| \geq
|\beta_1+\beta_3|+|\beta_2+\beta_4|\geq 2^{\ell+1}.
\]
If $[v]=0$ then we replace $v$ by $w$.
\end{proof}
\end{proof}

\begin{remarque}
  This proposition would be a direct application of question
  \ref{quest:ConjReid} if it were answered true.
  It is also an example of chain complexes product with minimum
  distance equal to the product of the minimum distances.
\end{remarque}

\begin{remarque}
  It is explicit in the proof that minimally weighted
  homology-surviving elements are carried by $4^\ell$ generators only,
  namely those of $C_1(\Dul{1})^{\otimes \ell}$.
\end{remarque}
\begin{question}
  Can unlink codes be swept out, for instance by removing acyclic subcomplexes, so they
  reach parameter $\llbracket4^\ell;2^\ell;2^\ell\rrbracket$ ?
  Since it would share almost the same dimension and same logarithmic sparseness
  property, would it be somehow related to Couvreur--Delfosse--Zemor
  codes (\cite{CDZ})~?
\end{question}

\subsection{Sparseness}
\label{sec:Sparseness-unlink}

\begin{prop}
  The weight of each row in the $\ell^\textrm{th}$ unlink code is
  $O\big(\ln(n_\ell)\big)$ as $\ell$ increases.
\end{prop}
\begin{proof}
  It is clear from Khovanov homology construction that each row has
  between $\ell+1$ and $2(\ell+1)$ non trivial entries.
  Since $4^\ell\leq n_\ell\leq 6^\ell$, the result follows.
\end{proof}



\section{$(2,n)$--torus link codes}
\label{sec:torus-codes}

For every $\ell\in\N$, we consider the following diagram $\Dtl{\ell}$ of the
pointed $(2,\ell)$--torus link:

\[
\dessin{2.5cm}{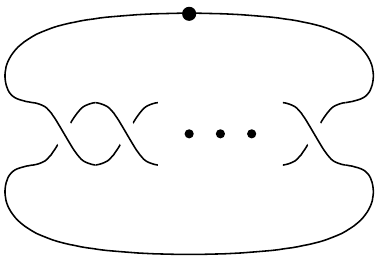}.
\]
For every $r\in\llbracket2,\ell\rrbracket$, the code obtained from
$\big(C^{r-1}(\Dtl{\ell})\xrightarrow{\p_{\Dtl{\ell}}}
C^r(\Dtl{\ell})\xrightarrow{\p_{\Dtl{\ell}}} C^{r+1}(\Dtl{\ell})\big)$ is called \emph{$(\ell,r)^\textrm{th}$ $(2,n)$--torus link code}.
Its parameters are denoted by $\llbracket n_{\ell,r};k_{\ell,r};d_{\ell,r}\rrbracket$.

\subsection{Homology}
\label{sec:homology}

For convenience, we introduce, for every $\ell\in\N$ the diagram
$U_\ell:=\dessin{1.2cm}{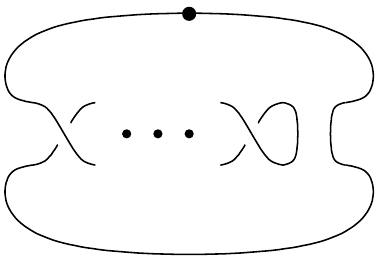}$.
It follows from Prop. \ref{prop:Invariance} that $\Kh(U_\ell)$ and
$\Kh(U_\ell!)$ have only one non-zero element, respectively in degree $\ell$
and 0.
Then the exact long sequence presented in section
\ref{sec:exact-sequence}, applied to the rightmost crossing, gives for every $\ell\in\N^*$
\[
\begin{array}{l}
\xymatrix{0\ar[r]&\Kh^r(\Dtl{\ell})\ar[r]^{\beta^r_\ell}&\Kh^r(\Dtl{\ell-1})\ar[r]&0}\textrm{
  for }r\in\llbracket0,\ell-2\rrbracket
  \\[.5cm]
  \xymatrix@R=.3cm{0\ar[r]&\Kh^{\ell-1}(\Dtl{\ell})\ar[r]^{\beta^{\ell-1}_\ell}&\Kh^{\ell-1}(\Dtl{\ell-1})\ar[r]&\Kh^{\ell-1}(U_{\ell-1})\ar[r]^{\alpha_\ell}\ar@{}[d]|{\rotatebox{270}{$\cong$}}&\Kh^\ell (\Dtl{\ell})\ar[r]&0\\
    &&&\F_2&&}\\[1.2cm]
  \xymatrix{0\ar[r]&\Kh^{r-1}(\Dtl{\ell-1}!)\ar[r]^{\alpha^r_\ell}&\Kh^r(\Dtl{\ell}!)\ar[r]&0}\textrm{
  for }r\in\llbracket2,\ell\rrbracket \\[.5cm]
\xymatrix@R=.3cm{0\ar[r]&\Kh^0(\Dtl{\ell}!)\ar[r]^{\beta_\ell}&\Kh^0(U_{\ell-1}!)\ar[r]\ar@{}[d]|{\rotatebox{270}{$\cong$}}&\Kh^0(\Dtl{\ell-1}!)\ar[r]^{\alpha^1_n}&\Kh^1(\Dtl{\ell}!)\ar[r]&0\\
    &&\F_2&&&}.
\end{array}
\]
But $\Kh^\ell (\Dtl{\ell})\neq 0$ since
$\func{\p_{\Dtl{\ell}}}{C^{\ell-1}(\Dtl{\ell})}{C^\ell (\Dtl{\ell})}$ involves only splitting
circles, so the weight of any image is necessarily even and every
single generator survives in homology.
Similarly, it is easy to produce a non trivial element in the kernel of
$\func{\p_{\Dtl{\ell}!}}{C^0(\Dtl{\ell}!)}{C^1(\Dtl{\ell}!)}$ and since there is nothing
to quotient by, it follows that $\Kh^0(\Dtl{\ell}!)\neq 0$.

Then, by induction, we can deduce that all the named maps are
isomorphisms and that:
\[
\Kh_r(\Dtl{\ell})=\left\{
  \begin{array}{l}
    \F_2\textrm{ for }r=0\textrm{ and }r\in\llbracket2,\ell\rrbracket\\
    0 \textrm{ otherwise}
  \end{array}
\right.
\hspace{2cm}
\Kh_r(\Dtl{\ell}!)=\left\{
  \begin{array}{l}
    \F_2\textrm{ for }r\in\llbracket0,\ell-2\rrbracket\textrm{ and }r=\ell\\
    0 \textrm{ otherwise}
  \end{array}
\right..
\]

\subsection{Length and dimension}
\label{sec:length-dimension-torus}

\begin{prop}\label{prop:TorusLength}
  $n_{\ell,r}=2^{r-1}{\ell\choose r}$ and $k_{\ell,r}=1$.
\end{prop}
\begin{proof}
  Concerning the length, one have to choose the $r$ 1--resolved crossings and then it remains
  $r-1$ undotted circles to label.

  The dimension has been computed in the previous section.
\end{proof}

\subsection{Minimum distance}
\label{sec:minimum-distance-torus}

\begin{prop}
 For $r\in\llbracket2,\ell\rrbracket$, $d^r_{\Dtl{\ell}}={\ell\choose r}$ and $d^0_{\Dtl{\ell}}=2$.
\end{prop}
\begin{proof}
  Within the framework of this proof and for simplicity, we will
  denote $\Dtl{\ell}$ by $D$ and $\{+,-\}^{\ell-1}$ by $S$.

Equality $d^0_{\Dtl{\ell}}=2$ follows from the fact that $C^0(D)$
has only two generators with equal non-zero image throught $\p_D$.

Now, we consider $r\in\llbracket2,\ell\rrbracket$.
  First we note that the cardinal of the set $E_r:=\Big\{\func{\phi}{\{\textrm{crossings
      of }D\}}{\{0,1\}}\ \Big|\ |\phi^{-1}(1)|=r\Big\}$ is
  ${\ell\choose r}$.
  Then, we construct a map
\[
\begin{array}{ccc}
    E_r & \longrightarrow &
    \{\textrm{labelling maps}\}\\[.2cm]
    \phi & \longmapsto & \sigma_\phi
  \end{array}
\]
\noi so that $\sum_{\phi\in E_r} D_\phi^{\sigma_\phi}$ is in the
kernel of $\func{\p_D}{C^r(D)}{C^{r+1}(D)}$.
  To this end, we choose $\ee:=(\e_1,\cdots,\e_{\ell-1})\in S$.\\
  When 1--resolving all the crossing of $D$, we obtain a resolution $D_\ee$
 with $\ell-1$ undotted
  circles and we label them, from left to right with
  $\e_1,\e_2,\cdots,\e_{\ell-2}$ and $\e_{\ell-1}$:
\[
D_\ee:=\dessin{2cm}{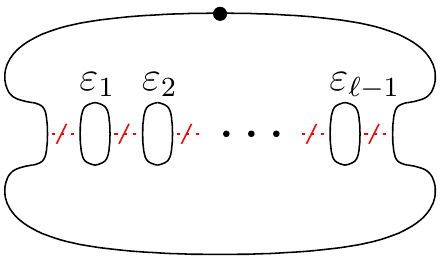}.
\]
To $\phi\in E_r$ corresponds a resolution of $D$ where $\ell-r$ crossings are
turned into $0$--resolutions.
Roughly, we define $D^r_{\ee}$ as the image of $D_\ee$ under the
partial maps $\p^{-1}_c$, for $c$ successively all these $\ell-r$ crossings.
Explicitly, the $\ell-1$ circles merge into $r-1$ ones and there are numbers
$a,b_1,\cdots,b_{r-1},c\in\N$ such that the first $a$ and the last
$c$ crossings are 0--resolved, and the $i^\textrm{th}$ circle, numbered from
left to right, contains $b_i$ 0--resolved crossings.
Note that the $(r+1)$-uple $(a,b_1,\cdots,b_{r-1},c)$ determines $\phi$.
Then we denote by $B_i$ the sum $1+a+\disp{\sum_{j=1}^{i-1}}(1+b_i)$ and
we define $\sigma_\phi$ the map which label the $i^\textrm{th}$ circle
by $\Lambda_i:=(-1)^{1+b_i}\e_{B_i}\e_{B_i+1}\cdots\e_{B_i+b_i}$.
For instance, the case $\ell=10$, $r=4$, $a=2$, $b_1=1$, $b_2=0$, $b_3=3$ and
$c=0$ gives
\[
\dessin{2cm}{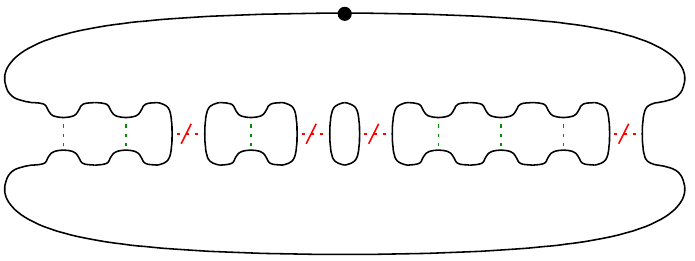}\ \leadsto\ \dessin{2cm}{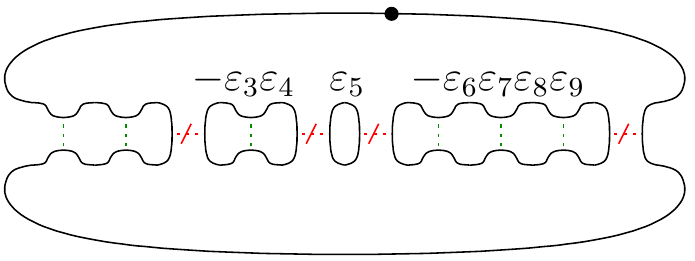}.
\]
We denote $\sum_{\phi\in E_r} D_\phi^{\sigma_\phi}$ by $D^r_\ee$ and claim that $\p_D(D^r_\ee)=0$.
It is sufficient to show that for any given $\varphi\in E_{r+1}$ the
contributions of the form $D_\varphi^\sigma$ cancel.
So let us choose such a $\varphi$.
As above, we can describe it by integers $a,b_1,\cdots,b_r,c$.
The elements of $E_r$ such that $\p_D(D_\phi^{\sigma_\phi})$
contributes are
\[
\begin{array}{cl}
  (a+1+b_1,b_2,\cdots,b_r,c);&\\
  (a,b_1,\cdots,b_{i-1},b_i+1+b_{i+1},b_{i+2},\cdots,b_r,c) & \textrm{
    for }i\in\llbracket1,r-1\rrbracket;\\
  (a,b_1,\cdots,b_{r-1},b_r+1+c). 
\end{array}
\]
The labels of their circles, given from left to right and using
the same notation $\Lambda_i$ as above are
\[
\begin{array}{cl}
 \big(\Lambda_2,\cdots,\Lambda_r\big); &\\
  \big(\Lambda_1,\cdots,\Lambda_{i-1},-\Lambda_i\Lambda_{i+1},\Lambda_{i+2},\cdots,\Lambda_r\big)& \textrm{
    for }i\in\llbracket1,r-1\rrbracket;\\
   \big(\Lambda_1,\cdots,\Lambda_{r-1}\big).
\end{array}
\]
And their contributions, again given by the labels of the circles, are
\[
\begin{array}{cl}
 \big(\Lambda_1,\Lambda_2,\cdots,\Lambda_r\big)\ ;\
 \big(-\Lambda_1,\Lambda_2,\cdots,\Lambda_r\big)\ &\\
  \big(\Lambda_1,\cdots,\Lambda_{i-1},-\Lambda_i,\Lambda_{i+1},\Lambda_{i+2},\cdots,\Lambda_r\big)\ ;\  \big(\Lambda_1,\cdots,\Lambda_{i-1},\Lambda_i,-\Lambda_{i+1},\Lambda_{i+2},\cdots,\Lambda_r\big)& \textrm{
    for }i\in\llbracket1,r-1\rrbracket;\\
   \big(\Lambda_1,\Lambda_2,\cdots,-\Lambda_r\big)\ ;\ \big(\Lambda_1,\Lambda_2,\cdots,\Lambda_r\big); 
\end{array}
\]
They do cancel indeed.
The element $D^r_\ee$ is hence an
element of the kernel of $\p_D$ which contains exactly one element for
each resolution of $E_r$.
But the map $\p_D$ only splits circles, so it produces even numbers of
contributions for each resolution of $E_r$.
So $D^r_\ee$ cannot be in the image of $\p_D$ and it survives in homology.
Moreover, it is of weight ${\ell\choose r}$.

Now, we assume {\it ad absurdum} that there exists $x\in\ker({\p_D}_{|C^r(D)})$, surviving in homology and
satisfying $|x|<|D^r_\ee|$.
Then there is a resolution of $E_r$ which
doesn't appear in $x$, and hence it appears exactly once in $x+D^r_\ee$.
It follows that $x+D^r_\ee$ survives in
homology for the same reason as $D^r_\ee$.
But $\dim\big(\Kh^r(D)\big)=1$, so $[x]=[D^r_\ee]$ and $[x+D^r_\ee]=0$.
This concludes the proof for $r\in\llbracket2,\ell\rrbracket$.
\end{proof}
\begin{remarque}
  Defining a representative of the non trivial homology class for
  every $(\e_1,\cdots,\e_{\ell-1})\in S$ is obviously
  redundant since the simpliest case, when all circles are
  labelled by $-$, would have been sufficient. However, all these $D^r_\ee$ will be helpful in the proof of  the
  next proposition.
\end{remarque}
\begin{prop}\label{prop:MinDist}
  For $r\in\llbracket0,\ell-2\rrbracket$, $d^r_{\Dtl{\ell}!}=2^{\ell-r-1}$ and $d^\ell_{\Dtl{\ell}!}=1$.
\end{prop}
\begin{proof}
  Since the homology is non trivial in degree $\ell$, the assertion on
  $d^\ell_{\Dtl{\ell}!}$ is also trivial.

  The map $\beta_\ell$ from section \ref{sec:homology} is an isomorphism
  and the map underlying $\beta_\ell$
  at the chain complexes level is a generator-preserving isomorphism.
  Since it follows from Prop. \ref{cor:ReidemeisterWeight} that
  $d^0_{U_{\ell-1}}=2^{\ell-1}$, Prop. \ref{prop:WeightMap} implies that
  $d^0_{\Dtl{\ell}!}= 2^{\ell-1}$.
  Then, inductive use of maps $\alpha^r_\ell$, for
  $r\in\llbracket1,\ell-2\rrbracket$, shows that
  $d^r_{\Dtl{\ell}!}\leq2^{\ell-r-1}$.

  Reciprocally, we consider an element $x\in\Ker(\p_{\Dtl{\ell}!})\cap
  C^r(\Dtl{\ell}!)$ such that $|x|<2^{\ell-r-1}$. Up to the reversing of all
  signs, $x$ can be seen as an element $x^\vee$ of the dual of
  $C^{\ell-r}(\Dtl{\ell}!)$.
  Using the notation
  of the previous proof, our goal is now to prove that there exists some
  $\ee\in S$ such that $x^\vee (D^r_\ee)=0$. It will follow
  from lemma \ref{lem:CoHoNull} and the fact that $\dim\big(C^{\ell-r}(\Dtl{\ell})\big)=1$, that $x^\vee$ is null in cohomology.
  Let $x_0$ be a generator of $C^{\ell-r}(\Dtl{\ell}!)$ and $x^\vee_0$ its dual
  element under the map $m$ of Prop. \ref{prop:Duality}. For $x^\vee_0(D^r_\ee)=0$ to hold, it is sufficient that $x_0$
  doesn't appear in $D^r_\ee$. The generator $x_0$ is determined by its
  labelling $(\eta_1,\cdots,\eta_{\ell-r-1})\in\{+,-\}^{\ell-r-1}$ read from left to
  right and its $r$ crossings $c$ which are
  $0$--resolved. It is easily checked that for each such crossing,
  $\p^{-1}_c(\textrm{any generator})$ contains exactly two
  elements. It follows that there are only $2^r$ elements
  $\ee\in S$ so that $x^\vee_0(D^r_\ee)\neq0$.
  As a consequence, there is at most $2^r|x|<2^{\ell-1}=\# S$ elements
  $\ee\in S$ so that $x^\vee (D^r_\ee)\neq0$. There is thus room for
  at least one $\ee\in S$ such that $x^\vee (D^r_\ee)=0$. This concludes the proof.
\end{proof}

\begin{cor}
  $d_{\ell,r}=\min\left\{{\ell\choose r},2^{r-1}\right\}$.
\end{cor}

\subsection{Summary}
\label{sec:summary}
\[
\begin{array}{|c|c|c|c|c|c|}
  \hline
    r \bigstrut & 0 & 1 & 2 & r\in\llbracket3,\ell-1\rrbracket  & \ell\\
    \hline
    \dim\Big(C^r(\Dtl{\ell})\Big) \rule[-.3cm]{0pt}{.8cm} & 2 & \ell & \ell (\ell+1) &
    2^{r-1}{\ell\choose r} &
    2^{\ell-1}\\
    \hline
    \dim\Big(\Kh^r(\Dtl{\ell})\Big) \rule[-.3cm]{0pt}{.8cm} & 1 & 0 & 1 & 1 & 
    1\\
    \hline
    d^r_{\Dtl{\ell}} \rule[-.3cm]{0pt}{.8cm} & 2 & \infty & \frac{\ell (\ell+1)}{2} &  
    {\ell\choose r} & 1\\
    \hline
    d^{\ell-r}_{\Dtl{\ell}!} \rule[-.3cm]{0pt}{.8cm} & 1 & \infty & 2 &  
    2^{r-1} & 2^{\ell-1}\\
    \hline
\end{array}
\]

\subsection{Extraction of a subfamily}
\label{sec:extraction-subfamily}

Since the minimum distance $d_{\ell,r}$ is a minimum involving ${\ell\choose
  r}$, it collapses for extremal values of $r$.
However, for $r\approx \frac{\ell}{2}$, we have, for large $\ell$, ${\ell\choose
  \frac{\ell}{2}}\sim \frac{2^{\ell+\frac{1}{2}}}{\sqrt{\pi \ell}}$ which is greater than
$2^{\frac{\ell}{2}-1}$.
So one can expect to find a ``best'' value $r_\ell$ such that ${\ell\choose
  r_\ell}\approx 2^{r_\ell-1}$.
As a matter of fact, for every $\ell\in\N^*$, we define
$r_\ell:=\arr\big(\alpha_0 \ell-\beta_0\ln(\ell)+\gamma_0\big)$ with $\arr(\ .\
)$
any
rounding function to the nearest integer, $\alpha_0$ the unique zero
in $(0,1)$ of the function $\big(x\mapsto (2x)^x(1-x)^{1-x}-1\big)$,
$\beta_0:=\frac{1}{2\ln\left(\frac{2\alpha_0}{1-\alpha_0}\right)}$ and
$\gamma_0:=\beta_0\ln\left(\frac{2}{\pi\alpha_0(1-\alpha_0)}\right)$.

\begin{prop}
  The family of $(\ell,r_\ell)^\textrm{th}$ $(2,\ell)$--torus link codes has
  asymptotical parameter $\llbracket n;1;d_n\rrbracket$ with $d_n>\frac{\sqrt{n}}{1,62}$.
\end{prop}
\begin{proof}
  This is a consequence of Prop. \ref{prop:BestParam}
\end{proof}

\begin{question}
  Computations suggests that the sequence $(2\e_\ell)_{\ell\in\N^*}$ is dense
  in $\left[-1,1\right]$.
  If true, or at least if there is a subsequence $(\e_{\ell_s})_{s\in\N}$
  converging to 0, then the subfamily of $(\ell_s,r_{\ell_s})^\textrm{th}$
  $(2,\ell)$--torus link codes would have an
  asymptotical parameter $\llbracket n;1;\sqrt{n}\rrbracket$ similar to
  Kitaev code one.
\end{question}

\subsection{Sparseness}
\label{sec:Sparseness-torus}

\begin{prop}
  If $(r_\ell)_{\ell\in\N}$ is any sequence satisfying $\alpha \ell\leq
  r_\ell\leq\beta \ell$ for every $\ell\in\N$ and some given
  $\alpha,\beta\in(0,1)$, then the weight of each row in the
  $(\ell,r_\ell)^\textrm{th}$ $(2,\ell)$--torus link code is
  $O\big(\ln(n_{\ell,r_\ell})\big)$ as $\ell$ increases.
\end{prop}
\begin{proof}
By construction, the rows of one matrix have exactly $2(\ell-r_\ell)$ non trivial
entries and the rows of the other matrix exactly $r_\ell$.
So the weight of each row is bounded below by
$\min\big(\alpha,2(1-\beta)\big)\ell$ and above by $\max\big(\beta,2(1-\alpha)\big)\ell$.
  But according Prop. \ref{prop:TorusLength}, the length is $2^{r_\ell-1}{\ell\choose r_\ell}\geq 2^{r_\ell-1}\geq \frac{(2^\alpha)^\ell}{2}$.
\end{proof}

\begin{remarque}
  The subfamily dicussed in the previous section
  satisfy such bounds for $r_\ell$.
\end{remarque}



\begin{appendix}
  \section{Technical proofs}
  \label{sec:open-question}

  For the sake of clarity, we gather in this appendix some analytical
  proofs which would have weight down the core of the text.

  \begin{prop}\label{prop:SumCarre}
    For any $x\in\R^*_+$, $\sum_{r=0}^\ell \left[{\ell\choose
        r}x^r\right]^2\sim\frac{(1+x)^{2 \ell+1}}{2\sqrt{x\pi \ell}}$ as $\ell$
    tends to infinity.
  \end{prop}
  \begin{proof}
    For every $\ell\in\N$, we define $\func{f_\ell}{\R}{\C}$ by
 $f_\ell (t)=(1+xe^{2it})^\ell=\sum_{r=0}^\ell{\ell\choose r}x^re^{2irt}$.
Then, since $f_\ell$ is clearly $\pi$--periodic and $L^2$, Parseval's identity gives

\[
\sum_{r=0}^\ell \left[{\ell\choose
    r}x^r\right]^2\ =\ \frac{1}{\pi}\int_{-\frac{\pi}{2}}^{\frac{\pi}{2}}|1+xe^{2it}|^{2\ell} dt
 =  \frac{1}{\pi}\int_{-\frac{\pi}{2}}^{\frac{\pi}{2}}
(1+2x\cos(2t)+x^2)^\ell dt = \int_I\frac{1}{\pi}e^{\ell f_x(t)}
\]
\noi with $I=\big[-\frac{\pi}{2},\frac{\pi}{2}\big]$ and $f_x(t)=\ln\big(1+2x\cos(2t)+x^2\big).$
Now, $f_x$ is smooth with $f'_x(t)=\frac{-4x\sin(2t)}{1+2x\cos(2t)+x^2}$ and
$f''_x(t)=\sin(2t)\times\textrm{something}-\frac{8x\cos(2t)}{1+2x\cos(2t)+x^2}$,
so 0 is the unique maximum of $f_x$ on $I$ and it is non
degenerate. It follows from the method of steepest descent that $\sum_{r=0}^\ell \left[{\ell\choose
    r}x^r\right]^2\sim
\frac{1}{\pi}\sqrt{\frac{2\pi}{\ell}}e^{\ell\ln(1+2x+x^2)}\frac{1}{\sqrt{\frac{8x}{1+2x+x^2}}}=\frac{(1+x)^{2 \ell+1}}{2\sqrt{x\pi
    \ell}}$.
  \end{proof}

  \begin{prop}\label{prop:Tn}
    $2^\ell\disp{\sum_{r=0}^\ell}\frac{{\ell\choose r}{2r\choose
    r}}{2^r}\sim \sqrt{\frac{3}{2\pi \ell}}6^\ell$ as $\ell$ tends to infinity.
  \end{prop}
  \begin{proof}
    We consider the power series $f(x)=\sum_{\ell\geq0}T_\ell x^\ell$
    with $T_\ell:=2^\ell\disp{\sum_{r=0}^\ell}\frac{{\ell\choose r}{2r\choose
    r}}{2^r}$.  This
    is well defined in a neighrborhood of 0 since $T_\ell$ is clearly bounded above by
    $\big((t^{-1}+4+t)^\ell\big)_{t=1}=6^\ell$.  Then, for $x$ sufficiently
    small, we have
    \[
    f(x)=\sum_{\ell\geq0}\sum_{r=0}^\ell 2^\ell\frac{{\ell\choose r}{2r\choose
        r}}{2^r}x^\ell=\sum_{r\geq0}\sum_{\ell\geq r}\frac{{2r\choose
        r}}{2^r}{\ell\choose r}(2x)^\ell.
    \]
    It is standard to check that $\frac{1}{\sqrt{1-4z}}=\sum_{r\geq
      0}{2r\choose r}z^r$ and, for any $r\in\N$,
    $\frac{z^r}{(1-z)^{r+1}}=\sum_{\ell\geq r}{\ell\choose r}z^\ell$.  So, we
    can deduce
    \[
    f(x) = \sum_{r\geq0}\frac{{2r\choose
        r}}{2^r}\frac{(2x)^{r}}{(1-2x)^{r+1}} =
    \frac{1}{1-2x}\sum_{r\geq0}{2r\choose
      r}\left(\frac{x}{1-2x}\right)^r =
    \frac{1}{(1-2x)\sqrt{1-4\frac{2x}{1-2x}}}=
    \frac{1}{\sqrt{1-8x+12x^2}}
    \]
    But since it is known (see {\it e.g.} \cite{Abramowitz}, formula
    22.9.1 for $\alpha=\beta=0$, p. 783) that
    $\frac{1}{\sqrt{1-2xt+t^2}}=\sum_{\ell\geq0}P_\ell (x)t^\ell$ where $P_\ell$ is
    the $\ell^\textrm{th}$ Legendre polynomial. It follows that
    $T_\ell=\left(2\sqrt{3}\right)^\ell P_\ell\left(\frac{2}{\sqrt{3}}\right)$.
    On the other hand, it is also known (see {\it e.g.}
    \cite{Abramowitz}, formula 22.3.1 for $\alpha=\beta=0$, p. 775)
    that $P_\ell (x)=\frac{1}{2^\ell}\disp{\sum_{r=0}^\ell} {\ell\choose
      r}^2(x-1)^{n-r}(x+1)^r$, so
    \[
    x_0^\ell T_\ell=\disp{\sum_{r=0}^\ell}\left[{\ell\choose r}x_0^r\right]^2
    \]
    \noi with
    $x_0=\frac{1}{2-\sqrt{3}}$.
    Using Prop. \ref{prop:SumCarre} for $x=x_0$, we obtain the desired
    $T_\ell\sim \sqrt{\frac{3}{2\pi \ell}}6^\ell$.
  \end{proof}

\begin{prop}\label{prop:BestParam}
  If, for every $\ell\in\N^*$, $r_\ell:=\arr\big(\alpha_0
  \ell-\beta_0\ln(\ell)+\gamma_0\big)$ with $\alpha_0$ the unique zero
in $(0,1)$ of the function $\big(x\mapsto (2x)^x(1-x)^{1-x}-1\big)$,
$\beta_0:=\frac{1}{2\ln\left(\frac{2\alpha_0}{1-\alpha_0}\right)}$ and
$\gamma_0:=\beta_0\ln\left(\frac{2}{\pi\alpha_0(1-\alpha_0)}\right)$,
then, for every sufficiantly large integer $\ell$, $\min\left\{{\ell\choose
      r_\ell},2^{r_\ell-1}\right\}> \frac{\sqrt{2^{r_\ell-1}{\ell\choose
      r_\ell}}}{1,62}$.
\end{prop}
\begin{proof}
  First we note that
$\frac{1}{2}+\beta_0\ln\left(\frac{1-\alpha_0}{2\alpha_0}\right)=0$
and that
$\left(\frac{2\alpha_0}{1-\alpha_0}\right)^{\gamma_0}=\sqrt{\frac{2}{\pi\alpha_0(1-\alpha_0)}}$.

Now, we write $r_\ell=\alpha_0 \ell-\beta_0\ln(\ell)+\gamma_0+\e_\ell$ with
  $|\e_\ell|\leq\frac{1}{2}$ and $\gamma_\ell:=\gamma_0+\e_\ell$.
  Stirling's approximation applied to $\ell!$, $r_\ell!$ and $(n-r_\ell)!$ gives
  \begin{eqnarray*}
    \frac{2^{r_\ell-1}}{{\ell\choose r_\ell}} & = &
    \sqrt{\frac{\ell\pi}{2}}\sqrt{\frac{r_\ell}{\ell}\left(1-\frac{r_\ell}{\ell}\right)}\left(\frac{2r_\ell}{\ell}\right)^{r_\ell}\left(1-\frac{r_\ell}{\ell}\right)^{\ell-r_\ell}\big(1+o(1)\big)\\
    & = & \sqrt{\frac{\ell\pi\alpha_0(1-\alpha_0)}{2}}\underbrace{\left(\frac{2r_\ell}{\ell}\right)^{r_\ell}\left(1-\frac{r_\ell}{\ell}\right)^{\ell-r_\ell}}_{A_\ell}\big(1+o(1)\big).
  \end{eqnarray*}
  Then
  \begin{eqnarray*}
    A_\ell & = &
    2^{\alpha_0-\beta_0\ln(\ell)+\gamma_\ell}\left(\alpha_0-\beta_0\frac{\ln(\ell)}{\ell}+\frac{\gamma_\ell}{\ell}\right)^{\alpha_0
    \ell-\beta_0\ln(\ell)+\gamma_\ell}\left(1-\alpha_0+\beta_0\frac{\ln(\ell)}{\ell}-\frac{\gamma_\ell}{\ell}\right)^{(1-\alpha_0) \ell+\beta_0\ln(\ell)-\gamma_\ell}\\
  & = &
  \left((2\alpha_0)^{\alpha_0}(1-\alpha_0)^{1-\alpha_0}\right)^\ell\frac{2^{\gamma_\ell}}{2^{\beta_0\ln(\ell)}}B_\ell C_\ell D_\ell\
  = \ \frac{2^{\gamma_\ell}}{\ell^{\beta_0\ln2}}B_\ell C_\ell D_\ell
  \end{eqnarray*}
\noi with
\[
\hspace{-1cm}
  B_\ell\ := \ \left(1-\beta_0\frac{\ln(\ell)}{\alpha_0
      \ell}+\frac{\gamma_\ell}{\alpha_0 \ell}\right)^{\alpha_0 \ell}\hspace{.5cm}
  C_\ell \ := \ \left(1+\beta_0\frac{\ln(\ell)}{(1-\alpha_0)
      \ell}-\frac{\gamma_\ell}{(1-\alpha_0) \ell}\right)^{(1-\alpha_0) \ell}\hspace{.5cm}
  D_\ell \ := \ \left(\frac{1-\alpha_0+\beta_0\frac{\ln(\ell)}{\ell}-\frac{\gamma_\ell}{\ell}}{\alpha_0-\beta_0\frac{\ln(\ell)}{\ell}+\frac{\gamma_\ell}{\ell}}\right)^{\beta_0\ln(\ell)-\gamma_\ell}.
\]
But
\[
  B_\ell =
  e^{\alpha_0 \ell\ln\left(1-\beta_0\frac{\ln(\ell)}{\alpha_0 \ell}+\frac{\gamma_\ell}{\alpha_0 \ell}\right)} = 
  e^{\alpha_0 \ell\left(-\beta_0\frac{\ln(\ell)}{\alpha_0 \ell}+\frac{\gamma_\ell}{\alpha_0 \ell}+o\left(\frac{1}{\ell}\right)\right)} =  e^{-\beta_0\ln(\ell)+\gamma_\ell+o(1)} = \ell^{-\beta_0}e^{\gamma_\ell}\big(1+o(1)\big),
\]
\noi and similarly
$C_\ell=\ell^{\beta_0}e^{-\gamma_\ell}\big(1+o(1)\big)$.
Concerning $D_\ell$, we have
\begin{eqnarray*}
  D_\ell & =&
  e^{\left(\bigstrut\beta_0\ln(\ell)-\gamma_\ell\right)\left(\ln\left(\frac{1-\alpha_0}{\alpha_0}\right)+{\large\strut}\ln\left(1+\beta_0\frac{\ln(\ell)}{(1-\alpha_0) \ell}-\frac{\gamma_\ell}{(1-\alpha_0) \ell}\right)+{\large\strut}\ln\left(1-\beta_0\frac{\ln(\ell)}{\alpha_0 \ell}+\frac{\gamma_\ell}{\alpha_0 \ell}\right)\right)}\ = \
e^{\left(\bigstrut\beta_0\ln(\ell)-\gamma_\ell\right)\left(\ln\left(\frac{1-\alpha_0}{\alpha_0}\right)+{\large\strut}o\left(\frac{\ln(\ell)}{\ell}\right)\right)}\\
 & = &
e^{\beta_0\ln\left(\frac{1-\alpha_0}{\alpha_0}\right)\ln(\ell)-\gamma_\ell\ln\left(\frac{1-\alpha_0}{\alpha_0}\right)+o(1)}\ = \ \ell^{\beta_0\ln\left(\frac{1-\alpha_0}{\alpha_0}\right)}\left(\frac {\alpha_0}{1-\alpha_0}\right)^{\gamma_\ell}\big(1+o(1)\big).
\end{eqnarray*}
Finally, we get
$A_\ell=\frac{2^{\gamma_\ell}}{\ell^{\beta_0\ln(2)}}\ell^{\beta_0\ln\left(\frac{1-\alpha_0}{\alpha_0}\right)}\left(\frac {\alpha_0}{1-\alpha_0}\right)^{\gamma_\ell}\big(1+o(1)\big)=\left(\frac{2\alpha_0}{1-\alpha_0}\right)^{\gamma_\ell}\ell^{\beta_0\ln\left(\frac{1-\alpha_0}{2\alpha_0}\right)}\big(1+o(1)\big)$
and hence
\begin{eqnarray*}
\frac{2^{r_\ell-1}}{{\ell\choose r_\ell}} & = &
\sqrt{\frac{\pi\alpha_0(1-\alpha_0)}{2}}\left(\frac{2\alpha_0}{1-\alpha_0}\right)^{\gamma_\ell}\ell^{\frac{1}{2}+\beta_0\ln\left(\frac{1-\alpha_0}{2\alpha_0}\right)}\big(1+o(1)\big)\\
& = &
\left(\frac{2\alpha_0}{1-\alpha_0}\right)^{\e_\ell}\sqrt{\frac{\pi\alpha_0(1-\alpha_0)}{2}}\left(\frac{2\alpha_0}{1-\alpha_0}\right)^{\gamma_0}\big(1+o(1)\big)\
=\ \left(\frac{2\alpha_0}{1-\alpha_0}\right)^{\e_\ell}+o(1).
\end{eqnarray*}
 Now, computations show that
$\delta_0:=2.61>\sqrt{\frac{2\alpha_0}{1-\alpha_0}}\geq\left(\frac{2\alpha_0}{1-\alpha_0}\right)^{\e_\ell}\geq
\sqrt{\frac{1-\alpha_0}{2\alpha_0}}>\delta_0^{-1}$.
So, for $\ell$ sufficiently large, we have $\delta_0^{-1}{\ell\choose r_\ell}
<2^{r_\ell-1}<\delta_0 {\ell\choose r_\ell}$.
For symmetry reason, we may assume that $2^{r_\ell-1}\leq{\ell\choose
  r_\ell}$. Then $(2^{r_\ell-1})^2>\frac{2^{r_\ell-1}{\ell\choose r_\ell}}{\delta_0}$
and the results follows since $\sqrt{\delta_0}<1.62$.
\end{proof}
\end{appendix}


\bibliographystyle{amsalpha}
\bibliography{KhovanovCodes}
\addcontentsline{toc}{part}{Bibliography}

\end{document}